\documentclass[11pt]{article}

\usepackage[utf8]{inputenc}
\usepackage{amsmath}
\usepackage{amsfonts}
\usepackage{mathtools}
\usepackage{graphicx}
\usepackage{color}
\usepackage[hidelinks]{hyperref}
\usepackage{tabularx}
\usepackage{float}
\usepackage{array}
\usepackage{multirow}
\usepackage{makecell}
\usepackage{url}
\usepackage{colortbl}
\usepackage{dsfont}
\usepackage{stmaryrd}
\usepackage{natbib}
\usepackage{caption}
\usepackage{color}
\usepackage{xcolor}
\usepackage{titlesec}
\usepackage{subcaption}
\usepackage{algorithm}
\usepackage{algpseudocode}
\usepackage{url}

\renewcommand{\baselinestretch}{1.3}

\makeatletter
\@addtoreset{equation}{section}

\makeatother

\newcounter{Fig}[figure]

\newcounter{Tab}[table]

  {\refstepcounter{Tab}
   \addtocounter{table}{1}
   \samepage\vspace{0.2cm}
   \centerline{#1} \nobreak
   \begin{verse}{Table \thesection.\arabic{Tab}:}
  }%
  {\end{verse} \vspace{0.2cm}}

\relax

\newcommand{\bqa}{\begin{eqnarray*}}
\newcommand{\eqa}{\end{eqnarray*}}
\newcommand{\bqan}{\begin{eqnarray}}
\newcommand{\eqan}{\end{eqnarray}}
\newcommand{\bqt}{\begin{quote}}
\newcommand{\eqt}{\end{quote}}
\newcommand{\bt}{\begin{tabbing}}
\newcommand{\et}{\end{tabbing}}
\newcommand{\bit}{\begin{itemize}}
\newcommand{\eit}{\end{itemize}}
\newcommand{\ben}{\begin{enumerate}}
\newcommand{\een}{\end{enumerate}}
\newcommand{\beq}{\begin{equation}}
\newcommand{\eeq}{\end{equation}}
\newcommand{\bdefi}{\begin{definition}}
\newcommand{\edefi}{\end{definition}}
\newcommand{\bpro}{\begin{proposition}}
\newcommand{\epro}{\end{proposition}}
\newcommand{\blem}{\begin{lemma}}
\newcommand{\elem}{\end{lemma}}
\newcommand{\bth}{\begin{theorem}}
\newcommand{\eth}{\end{theorem}}
\newcommand{\bco}{\begin{corollary}}
\newcommand{\eco}{\end{corollary}}
\newcommand{\bdes}{\begin{description}}
\newcommand{\edes}{\end{description}}
\newcommand{\bre}{\begin{remark}}
\newcommand{\ere}{\end{remark}}

\newtheorem{definition}{Definition}[section]
\newtheorem{proposition}[definition]{Proposition}
\newtheorem{lemma}[definition]{Lemma}
\newtheorem{theorem}[definition]{Theorem}
\newtheorem{corollary}[definition]{Corollary}
\newtheorem{remark}[definition]{Remark}

\newenvironment{proof}[1][Proof]{\textbf{#1.} }{\ \rule{0.5em}{0.5em}}

\begin{document}

\begin{titlepage}

\title{Index insurance under demand and solvency constraints}

\author{{\large Olivier L\textsc{opez}$^1$},
{\large Daniel N\textsc{kameni}$^{1,2}$}  }

\date{\today}
\maketitle

\renewcommand{\baselinestretch}{1.1}

\begin{abstract}
Index insurance is often proposed to reduce protection gaps, especially for emerging risks. Unlike traditional insurance, it bases compensation on a measurable index, enabling faster payouts and lower claim management costs. This approach benefits both policyholders, through quick payments, and insurers, through reduced costs and better risk control due to reliable data and robust statistical estimates. An important difference with the concept of Cat Bonds is that the feasibility of such coverage relies on the possibility of mutualization. Mutualization, in turn, is achieved only if a sufficiently high number of policyholders agree to subscribe. The purpose of this paper is to introduce a model for the demand for index insurance and to provide conditions under which the solvency of the portfolio is achieved. From these conditions, we deduce a product that combines index and traditional indemnity insurance in order to benefit from the best of both approaches. We illustrate our results with a practical example involving the design of an index insurance product in the field of cyber insurance.
\end{abstract}

\vspace*{0.5cm}


\vspace*{0.5cm}


\vspace*{0.5cm}

\noindent{\bf Short title:} Demand for index insurance.

\vspace*{.5cm}

{\small
\parindent 0cm
$^1$ CREST Laboratory, CNRS, Groupe des Écoles Nationales d'Économie et Statistique, Ecole Polytechnique, Institut Polytechnique de Paris, 5 avenue Henry Le Chatelier 91120 PALAISEAU, France\\
$^2$ Detralytics, 1-7 Cours Valmy 92923 Paris La Défense Cedex, France. \\

E-mails: olivier.lopez@ensae.fr, daniel.nkameni@ensae.fr}

\end{titlepage}

\small
\normalsize
\addtocounter{page}{1}

\section{Introduction}

Index insurance is often promoted as a solution capable of addressing certain structural weaknesses of traditional indemnity-based insurance (see for example \cite{barnett2007weather}, \cite{carter2017index}, \cite{prokopchuk2018index} or \cite{han2019weather}). The principle of these coverages lies in the use of an index (or parameter) that can be easily calculated based on information available immediately after the incident. This automated calculation significantly simplifies claim management. Since the compensation is determined without the need for an expert to evaluate the amount, it can be paid out very quickly to the policyholder. The clarity of the indemnification conditions also reduces the likelihood of legal disputes. On the other hand, the policyholder must bear basis risk (see \cite{clement2018global}), as the compensation is not based on the actual loss but instead on an approximation of the latter. Consequently, there is a concern that index insurance may have a disappointing aspect, which is seen as a barrier to its development (see \cite{johnson2021paying}).

The design of an index that can serve as the basis for insurance coverage is similar to a statistical problem of estimation or prediction: using available variables, the goal is to approximate as closely as possible an unobserved quantity (the loss experienced by the policyholder), see \cite{nhess-21-2379-2021}. However, this problem involves a number of constraints due to the need to align with the policyholder's expectations. For instance, the construction of an index proposed by \cite{conradt2015flexible}, \cite{zhang2019index}, or \cite{chen2023managing} is based on maximizing the policyholder's utility rather than relying on a more standard metric commonly used in regression or forecasting.

Indeed, behind the question of meeting the policyholder's needs lies the issue of demand. Index insurance remains an insurance product: while other index-based products, such as Cat Bonds, can achieve balance through diversification strategies inherent to financial instruments, most index insurance products rely on mutualization to withstand adverse outcomes. However, the mutualization mechanism requires a sufficiently large number of policyholders. Insufficient demand will weaken the product, beyond simply failing to recover the cost of designing the index.

The purpose of the present paper is to study, through a modeling of the demand for index insurance, the viability of such a product when it competes with a traditional indemnity-based insurance product. The aim here is less about designing an optimal index and more about examining the conditions under which such an index becomes acceptable in a situation where the insurer must meet a solvency requirement. Particular attention will be paid to the impact of the loading factor applied to the premium (this loading factor can be lower in the case of index insurance due to reduced management costs), policyholders' risk aversion, and compensation delays of the competing indemnity-based insurance as levers for achieving this objective. This analysis will also lead us to propose the construction of a ``hybrid" coverage: by hybrid, we mean a combination of traditional indemnity-based insurance with index insurance, which supplements it in certain cases to leverage the best of both approaches, with index insurance not intervening when the basis risk is too high.

The rest of the paper is organized as follows. In Section \ref{sec:notations}, we set up the notations and formalize the problem of demand in index insurance and its link with solvency. Section \ref{sec:conditions} then shows conditions to meet (under the assumption that the utility of the policyholder is exponential) in order to achieve a sufficient demand in index insurance. These conditions lead to a natural choice for a hybrid coverage, mixing traditional and index insurance. This hybrid coverage is presented in Section \ref{sec:hybrid}. A practical example in cyber insurance is then provided in Section \ref{sec:practical}.

\section{Notation and settings}
\label{sec:notations}

In this section, we formalize the general framework that we consider to study the demand (Section \ref{sec:demand}) and solvency (Section \ref{sec:lln}) of a portfolio of index insurance products—first in the classical case where policyholders are independent, and then by introducing Section \ref{sec:system} to add an accumulation component that may arise when a large number of policyholders experience a simultaneous claim.

\subsection{Index insurance demand and statement of the problem}
\label{sec:demand}

We consider a situation where a policyholder has a choice between two insurance products: 
\begin{itemize} \item The first one is a ``traditional" indemnity-based insurance product, where the loss $Y \geq 0$ of the policyholder in the following year is fully covered. \item The second one is an index insurance product, which is based on $\phi(\mathbf{W})$, where $\mathbf{W}$ is a set of covariates measured after a claim in order to compute the index. \end{itemize}

Typically, $\phi(\mathbf{W})$ will be lower than $Y$, with overcompensation assumed to be rare in index insurance. To simplify, we consider that a policyholder experiences at most one claim; otherwise, $\mathbf{W}$ should be understood as the (different) circumstances of all encountered claims. If there is no claim, $Y = 0$ and $\phi(\mathbf{W}) = 0$.

A traditional way to model the demand in insurance relies on the concept of expected utility, as in \cite{cummins2004demand}, \cite{hao2018insurance}, or \cite{eeckhoudt1992background}. A refinement of this approach, especially in the context where there is a choice between different products, is proposed in \cite{braun2004impact} or \cite{fujii2016regret} and relies on the notion of ``regret". In this paper, we concentrate on the most classical framework, that is, expected utility, also because of recent contributions to the design of index-based covers which rely on expected utility maximization (see, for example, \cite{zhang2019index} or \cite{chen2023managing}).

Consider a class of utility functions $\mathcal{U}=\left\{x\in \mathbb{R} \rightarrow U_{\alpha}(x):\alpha \in \mathcal{A}\right\},$ where $\mathcal{A}\subset \mathbb{R}^k.$ Every function $U_{\alpha}$ is assumed to be non-decreasing and strictly concave to represent risk aversion. Each policyholder is associated with a different function, that is, with a different value of $\alpha.$ The decision to buy index insurance compared to indemnity-based insurance is motivated by the maximization of the expected utility of the final wealth.

\begin{itemize}
\item In the case of index insurance, the policyholder pays a premium $\pi_{\phi},$ incurs a loss $Y$, and receives a compensation $\phi(\mathbf{W}).$ The corresponding expected utility is
$$\frak{U}_{\phi}(\alpha)=E\left[U_{\alpha}(\phi(\mathbf{W})-Y-\pi_{\phi})\right].$$
where $\pi_{\phi}=(1+\theta)E[\phi(W)],$ and $\theta>0$ is the loading factor of the index insurance product.
\item In the case of indemnity-based insurance, the difference lies in the fact that the price is $\pi_{Y},$ but the compensation is $Y.$ However, we want to reflect the fact that this compensation is usually paid with a longer delay than with index insurance, which may be problematic for the policyholder who requires liquidity to repair the damages. Therefore, we consider that the compensation will be discounted by a factor $\exp(-\tau)$ for some $\tau\geq 0.$ This leads to the following expected utility for this solution:
$$\frak{U}_{Y,\tau}(\alpha)=E\left[U_{\alpha}\left(\{\exp(-\tau)-1\}Y-\pi_Y\right)\right].$$
where $\pi_{Y}=(1+\theta_Y)E[Y],$ and $\theta_Y>0$ is the loading factor of the indemnity-based insurance product.
\end{itemize}

Let us note that we did not consider the initial capital of the policyholder. This can be taken into account by incorporating it into the parameter $\alpha$ (which may be multivariate). Moreover, let us assume that we consider a function defined on $\mathbb{R}$ to allow the possibility of a negative value for the policyholder's wealth (in which case, a debt is created).

A policyholder with parameter $\alpha\in \mathbb{R}^k$ will choose to rely on index insurance only if

\begin{equation}\label{cond_achat}\frak{U}_{\phi}(\alpha)-\frak{U}_{Y,\tau}(\alpha)>0.
\end{equation}

Here, we assume to simplify that the customer buys one of the two contracts. Alternatively, one could easily consider the option where a third choice of not buying any insurance protection is possible.

\begin{remark}Note that we consider the case in which the loss is completely covered by the indemnity-based contract. We consider this case for simplicity, but the results can be easily extended to other forms of indemnities. For example, in the case of proportional insurance, only a proportion $p$ of the claim is paid, leading to a compensation $pY$. The value of $p$ can be included in the definition of $\tau$, so that $\exp(-\tau)$ reflects not only the depreciation due to time, but also partial compensation. The loading factor of the premium, as stated above, would also include this proportion $p$. Another important case concerns the inclusion of a deductible, in which case the compensation becomes $(Y-d)_+$ for some $d>0$. In this case, the utility associated with the indemnity-based contract is smaller than the current definition of $\mathfrak{U}_{Y,\tau}(\alpha)$, so condition (\ref{cond_achat}) will tend to slightly underestimate the demand for parametric insurance, therefore providing a ``cautious'' bound if the objective is to decide whether the development of parametric insurance is relevant or not.\end{remark}

\subsection{Effect on mutualization}
\label{sec:lln}

To be viable, an index insurance product needs to be subscribed by a sufficient number of policyholders. Consider a target population of potential customers of large size $N$, the number of policyholders buying the index insurance contract will be approximately

\begin{equation}
n = N\int \mathbf{1}_{\frak{U}_{\phi}(\alpha)-\frak{U}_{Y,\tau}(\alpha)>0}d\mu(\alpha), \label{nombre}
\end{equation} 
where $\mu$ is the cumulative distribution of $\alpha$ among the population.

The question is then to know if this number $n$ is large enough to build a portfolio that is economically viable. Considering that $(Y_i)_{1\leq i \leq n}$ are the losses of the $n$ policyholders, the loss of the insurance company is

$$L_n(\pi_{\phi})=\sum_{i=1}^n \phi(\mathbf{W}_i)-n\pi_{\phi}.$$

The size of the portfolio should be large enough to ensure that ruin during the next year is avoided with a sufficiently high probability, that is we want

\begin{equation}
\mathbb{P}\left(L_n(\pi_{\phi})\geq 0\right)\leq \varepsilon,  \label{solvency}
\end{equation}
where $\varepsilon>0$ is close to zero.

If we assume that the policyholders are independent and identically distributed, the Central Limit Theorem applies, and
$$n^{1/2}\left\{\frac{1}{n}\sum_{i=1}^n \phi(\mathbf{W}_i)-\pi^*_{\phi}\right\}\underset{n \to \infty}{\xrightarrow{}} \mathcal{N}\left(0,\sigma^2_{\phi}\right),$$
where $\sigma^2_{\phi}= Var(\phi(\mathbf{W}))$ and $\pi^*_{\phi}=E[ \phi(\mathbf{W})].$ From this distributional convergence, one can deduce the approximation

$$\mathbb{P}\left(L_n(\pi_{\phi})\geq 0\right)\approx S\left(\frac{n^{1/2}\theta \pi^*_{\phi}}{\sigma_{\phi}}\right),$$
where $S$ is the survival function of a $\mathcal{N}(0,1)$ variable, and $\pi_{\phi}=(1+\theta)\pi^*_{\phi},$ with $\theta>0$ being the loading factor of the index insurance product. Hence (\ref{nombre}) approximately rewrites

\begin{equation}
\label{nombre2}
\frac{n^{1/2}\theta \pi^*_{\phi}}{\sigma_{\phi}}\geq S^{-1}(\varepsilon).
\end{equation}

Of course, increasing $\theta$ does not necessarily lead to an improvement in solvency, since $n$ decreases with $\theta$ as demand is reduced.

We will keep this Gaussian approximation in the next parts, but let us note that this requires the variance of $\phi(\mathbf{W})$ to be finite, which may not be the case for extremely heavy-tailed distributions. If $\phi(\mathbf{W})$ is extremely heavy-tailed, and if $\varepsilon$ is small compared to $\pi_\phi$, other kinds of approximations based on Generalized Pareto distributions may be used, see for example \cite{mikosch} for more details. However, extremely heavy-tailed variables are not our main focus in the present paper, since they are incompatible with the exponential utility approach developed in section \ref{sec:notations} (which requires the loss to have a finite Laplace transform).

The independence and identical distribution assumption is justified if the field of application of the index insurance product is restricted to cases where policyholders are not individuals but large geographical areas (such as villages, cities, departments, regions, etc.) or large entities such as companies in the case of cyber insurance, as is the application in this paper. However, despite this restrictive context of application, there remains a slight probability of accumulation, at least on a fraction of the portfolio. For example, in the case of crop insurance (which is probably one of the most famous use cases of index products, see for example \cite{nhess-21-2379-2021} or \cite{barnett2007weather}), weather events may strike a group of geographical areas simultaneously. To account for that, we propose, in the following section, a simplified way to proceed via the introduction of a shock on the i.i.d. model that materializes the presence of catastrophic events.

\subsection{A simplified way to include accumulation phenomena}
\label{sec:system}

An accumulation phenomenon occurs when a significant number of policyholders encounter claims in a short period of time. In climate-related risk, this is essentially linked to the proximity between policyholders: policyholders that live in the same area are affected by similar weather conditions. In other cases, the dependence may not solely rely on geographic proximity, such as in cyber insurance. In such cases, it may be hopeless to obtain a clear map of the links between the policyholders. Moreover, trying to model the dependence between policyholders would introduce an additional difficulty in our context, where we take demand into account: even in the case of geographic dependence, this would require modeling a link between the distribution $\mu$ (describing the behaviors of the policyholders) and the localization of the potential customers. Calibrating the model would then require an important amount of data that may be difficult to obtain.

Consequently, we consider a simplified case where the accumulation episode materializes via an additional loss, which is heavy-tailed. The total loss of the portfolio is then

$$L_n(\pi_{\phi})=A_n+\sum_{i=1}^n \varrho_i\phi(\mathbf{W}_i)-n\pi_{\phi},$$
where $\varrho_i=0$ if policyholder $i$ was part of an accumulation episode, and $A_n$ represents the aggregated amount related to accumulation episodes.

For $A_n$, we consider a Generalized Pareto distribution, that is

$$\mathbb{P}(A_n\geq t)=\frac{1}{\left(1+\frac{\gamma t}{n s}\right)^{1/\gamma}},$$
with $\gamma<1$ and $s>0.$ Here we consider a Generalized Pareto where the scale parameter is proportional to $n$: this is the idea that the cost of the accumulation episode is proportional to the size of the portfolio.

In this case, the probability of ruin is bounded by

$$\mathbb{P}\left(L_n(\pi_{\phi})\geq 0\right)\leq \mathbb{P}\left(A_n-\frac{n\theta}{a}\geq 0 \right)+\mathbb{P}\left(\sum_{i=1}^n \phi(\mathbf{W}_i)-\left(\frac{1}{\theta}+\left\{1-\frac{1}{a}\right\}\right)n\theta\pi^*_{\phi}\geq 0\right),$$
for all $a>1.$ Using the same Gaussian approximation as in section \ref{sec:lln}, the right-hand side is approximately

\begin{eqnarray*}
 \frac{1}{\left(1+\frac{\gamma \theta}{a s}\right)^{1/\gamma}}+S\left((a-1)\frac{n^{1/2}\theta\pi^*_{\phi}}{a\sigma_{\phi}}\right).
\end{eqnarray*}
To make this quantity less than the tolerance $\varepsilon,$ we need

\begin{equation}
\label{cond2}
\frac{ n^{1/2}\theta\pi^*_{\phi}}{\sigma_{\phi}}\geq \frac{a}{(a-1)}S^{-1}\left(\varepsilon - \frac{1}{\left(1+\frac{\gamma \theta}{a s}\right)^{1/\gamma}}\right),
\end{equation}
if 

$$1<a<\frac{\gamma \theta \varepsilon^{\gamma}}{s(1-\varepsilon^{\gamma})},$$
which imposes that $\theta$ should be large enough to absorb the accumulation episode.

Logically, dealing with an additional accumulation risk increases the number of required policyholders to achieve mutualization. Moreover, a constraint appears on the loading factor $\theta$, which should be large enough. Again, since the achievable $n$ tends to decrease when the loading factor increases due to a lower demand, this number may become impossible to reach in some cases.

In the following section, we discuss conditions on the demand for (\ref{nombre2}) and (\ref{cond2}) to hold in the special case where the utility function is exponential: a simple form of the utility which allows obtaining light constraints on the measure $\mu.$

\section{Sufficient conditions for the viability of an index insurance product under exponential utility}
\label{sec:conditions}

In this section, we consider the particular case where the utility function is exponential. This allows us to simplify considerably the formulation of the problem and to provide, in section \ref{sec:exputil}, sufficient conditions for an index insurance product to be preferable compared to a traditional one. Then, we consider in section \ref{sec:expsolvency} the consequences on the solvency of the portfolio. Section \ref{sec:hybrid} introduces a way to combine traditional and index insurance to optimize the attractiveness of the product.

\subsection{Exponential utility}
\label{sec:exputil}

In this section, we consider $U_{\alpha}(x)=-(1/\alpha)\exp(-\alpha x)$. The parameter $\alpha$ can be interpreted as a materialization of risk aversion, in the sense that a policyholder with a high value of $\alpha$ will tend to accept a higher premium in exchange for insurance protection against the risk. In this case, the condition (\ref{cond_achat}) can be simplified. We introduce the Laplace transform and conditional Laplace transforms of $Y$,

\begin{eqnarray*}
\Psi_Y(\alpha) &=& E\left[\exp(\alpha Y)\right], \\
\psi_Y(\alpha|\mathbf{w}) &=& E\left[\exp(\alpha Y)|\mathbf{W}=\mathbf{w}\right].
\end{eqnarray*}

We assume that $\Psi_Y(\alpha)<\infty$ for all $\alpha$ in the support of $\mu$. Then,

\begin{eqnarray}
\frak{U}_{\phi}(\alpha)-\frak{U}_{Y,\tau}(\alpha)>0 \; & \Longleftrightarrow & \; -\alpha \left\{E\left[\psi_Y(\alpha|\mathbf{W})\exp\left(-\alpha \{\phi(\mathbf{W})-\pi_{\phi}\}\right)\right]-\Psi_Y(\alpha')\exp(\alpha \pi_Y)\right\}>0, \nonumber \\
& \Longleftrightarrow & \; E\left[\psi_Y(\alpha|\mathbf{W})\exp\left(-\alpha \phi(\mathbf{W})\right)\right]<\Psi_Y(\alpha')\exp\left(\alpha (\pi_Y-\pi_{\phi})\right), \label{cond_exp}
\end{eqnarray}
where $\alpha'=(1-\exp(-\tau))\alpha$.

From this expression, we see that condition (\ref{cond_achat}) is essentially a matter of bounding the difference $m_{Y}(\alpha|\mathbf{W})-\phi(\mathbf{W})$, where $m_Y(\alpha|\mathbf{w})=\log \psi_Y(\alpha|\mathbf{w})/\alpha$. The difference (which is positive from Jensen's inequality) should be small enough compared to the difference in prices $\pi_Y-\pi_{\phi}$, and the presence of $\Psi_Y(\alpha')$ allows this difference to go higher when $\tau$ increases.

Therefore, a first idea could be to take $\phi(\mathbf{w})=m_{Y}(\alpha|\mathbf{w})$ as an indemnity function. But this solution would not be efficient, in the sense that this pay-off would lead to too high a price: from Jensen's inequality, the pure premium would then be

$$E\left[m_{Y}(\alpha|\mathbf{w})\right]>E[Y].$$

Hence, except if we are very close to equality (which would happen only if $\alpha$ is close to zero and/or $Y|\mathbf{W}$ has a variance close to zero), it would become very difficult to offer a premium $\pi_{\phi}$ smaller than $\pi_Y$, and in some cases even impossible with a loading factor $\theta_{\phi}>0$. In addition to the problem of high prices, this type of contract would typically lead to a too important compensation in many cases. This may collide with some legal constraints depending on insurance regulations\footnote{For example, according to French legislation, this could be interpreted as ``enrichment without cause", although the existence of a claim may be a protection against this argument. See for example S. Bros, L’assurance paramétrique en assurance de dommages, bjda.fr 2023, Dossier n° 6.}.

To be compatible with this operational necessity to keep a low price, we consider a pay-off $\phi(\mathbf{w})=\phi_{\beta}(\mathbf{w})=\beta E[Y|\mathbf{W}=\mathbf{w}]$, with $\beta\leq 1$. This choice also allows us to control the probability of over-compensating for a claim. From Chernoff's inequality, this probability is

\begin{equation}\label{chernoff}
\mathbb{P}\left(Y-\beta E[Y|\mathbf{W}]<0|\mathbf{W}=\mathbf{w}\right)\leq \psi_Y(\rho|\mathbf{w})\exp\left(-\rho\beta E[Y|\mathbf{W}=\mathbf{w}]\right),
\end{equation}
for all $\rho>0$ such that $\psi_Y(\rho|\mathbf{w})<\infty$. This inequality can be used to control the proportion of cases where the compensation is too high.

Another approach to define the pay-off is to perform utility maximization as in \cite{zhang2019index} and \cite{chen2023managing}. In such an approach, the idea is to maximize a trade-off between large amounts of compensation $\phi(\mathbf{W})$ and an affordable price. However, let us note that, with this approach, $E[\phi(\mathbf{W})]$ may be greater than $E[Y]$ in some situations. Therefore, we prefer in this setting to consider a particular shape of pay-off, determining which difference of price between the indemnity and the index product is acceptable by the customer. If we determine a situation where $\phi(\mathbf{W})=\beta E[Y|\mathbf{W}]$, associated with a loading factor $\theta$, is an acceptable index product, an additional optimization can be performed if the constraint on $E[\phi(\mathbf{W})]\leq E[Y]$ is not an issue. Finally, we believe that the choice of the value of $\beta$ is best left in the hands of the insurer, as long as it is clearly stated at the signing of the contract and the policyholder is aware of the structure of the product being purchased. Indeed, one could imagine a product in which an insurer offers various proportions of coverage of the expected average loss, and policyholders choose the proportions corresponding to the premiums they are willing and able to pay. Such an approach, along with its advantages and limitations, is discussed to some extent in \cite{liang2023optimal}, \cite{marone2022should}, and \cite{woodard2016crop}.

Proposition \ref{prop_existence} shows that the index insurance product is chosen by a policyholder with risk aversion $\alpha$ provided that the loading factor $\theta$ is small enough and that constraint (\ref{cond_bound}) on the conditional Laplace transform at point $\alpha$ holds.

\begin{proposition}
\label{prop_existence}

Assume that

\begin{equation}
\label{cond_bound}
\sup_{\mathbf{w}\in \mathcal{W}}\frac{m_Y(\alpha|\mathbf{w})-\phi_{\beta}(\mathbf{w})}{E[Y]}< 1-\beta+\theta_Y.
\end{equation}
Let 

$$\eta(\alpha,\beta)=1-\beta+\theta_Y-\left\{\sup_{\mathbf{w}\in \mathcal{W}}\frac{m_Y(\alpha|\mathbf{w})-\phi_{\beta}(\mathbf{w})}{E[Y]}\right\}.$$
Then, for $\tau\geq 0,$ Condition (\ref{cond_achat}) holds if

$$\theta\leq \frac{\eta(\alpha, \beta)}{\beta}.$$
Hence, in this case, there exists an index insurance product with a positive loading factor that is preferable for a policyholder with risk aversion $\alpha.$
\end{proposition}
The proof is given in section \ref{sec_append_proof1}.

In this result, condition \ref{cond_bound} is key and needs to be examined more closely. By Jensen's inequality, $m_Y(\alpha|\mathbf{w})\geq E[Y|\mathbf{W}=\mathbf{w}],$ so 

$$m_Y(\alpha|\mathbf{w})-\phi_{\beta}(\mathbf{w})\geq (1-\beta)E\left[Y|\mathbf{W}=\mathbf{w}\right].$$ 
Taking the expectation, we see that 

\begin{equation}
\label{desc}\frac{E\left[m_Y(\alpha|\mathbf{w})-\phi_{\beta}(\mathbf{w})\right]}{E[Y]}\geq 1-\beta.
\end{equation}

Hence, a necessary condition for condition \ref{cond_bound} is that the difference between the left-hand side and right-hand side of (\ref{desc}) is not too high (less than $\theta_Y$). This difference tends to become lower when $\alpha$ tends to zero, confirming the intuition that risk aversion plays against the index insurance product. Moreover, a smaller value of $Var(Y|\mathbf{W}=\mathbf{w})$ will also reduce the gap between $m_Y(\alpha|\mathbf{w})-\phi_{\beta}(\mathbf{w})$ and its lower bound $(1-\beta)E[Y|\mathbf{W}=\mathbf{w}].$ This directly refers to the ability to efficiently predict $Y$ from the available information $\mathbf{W},$ from which is computed the index.

\subsection{Consequences on the solvency of the portfolio}
\label{sec:expsolvency}

Let us note that proposition \ref{prop_existence} only provides a sufficient condition for Condition (\ref{cond_achat}) to hold. It is valid for all values of $\tau,$ including $\tau=0,$ and is then sufficient to obtain a lower bound for the demand, since higher values of $\tau$ will increase the disadvantage of the indemnity-based insurance product. In case of $\tau>0,$ there is room for weakening (\ref{cond_bound}) and/or the condition on $\theta.$

From proposition \ref{prop_existence}, we easily get the following corollary.

\begin{corollary}
\label{cor_extension}
Assume that (\ref{cond_bound}) holds for some $\alpha_0>0.$ Then, if $\theta\leq \eta(\alpha, \beta)\beta^{-1},$ Condition (\ref{cond_achat}) holds for all values of $\alpha \in (0,\alpha_0+h_{\beta}(\tau)]$ with
$$h_{\beta}(\tau)=F^{-1}\left(F(\alpha_0)\exp(-\tau)\exp(-\alpha[\pi_Y-\pi_{\phi_{\beta}}])\right)-\alpha_0,$$
where $F(\alpha)=\Psi'(\alpha)=E[Y\exp(\alpha Y)],$ and

\begin{equation}
\label{lower_bound}
n \geq N \mu\left((0,\alpha_0+h_{\beta}(\tau)]\right).
\end{equation}
\end{corollary}

This result is a direct consequence of the more general result of Lemma \ref{lemma_heredity}.

\begin{proposition}
\label{prop_solv1}
Under the framework of section \ref{sec:lln} (i.i.d. policyholders and no accumulation phenomenon) and the conditions of Corollary \ref{cor_extension}, condition (\ref{solvency}) holds for $\phi_{\beta}$ provided that $\theta\leq \eta(\alpha, \beta)\beta^{-1}$ and

\begin{equation}
\label{magic}\eta(\alpha, \beta) \geq \frac{\sigma \beta 
 S^{-1}(\varepsilon)}{N^{1/2} E[Y] \mu\left((0,\alpha_0+h_{\beta}(\tau)]\right)^{1/2}},
\end{equation}

where

$$\sigma^2=Var\left(E[Y|\mathbf{W}]\right),$$
which can also be rewritten as

$$\sup_{\mathbf{w}\in \mathcal{W}}\frac{m_Y(\alpha|\mathbf{w})-\phi_{\beta}(\mathbf{w})}{E[Y]}\leq 1-\beta+\theta_Y-\frac{\sigma \beta  S^{-1}(\varepsilon)}{N^{1/2} E[Y] \mu\left((0,\alpha_0+h_{\beta}(\tau)]\right)^{1/2}}.$$
\end{proposition}

In the case where the possibility of accumulation episodes is included as in section \ref{sec:system}, the result of Proposition \ref{prop_solv2} is slightly modified.

\begin{proposition}
\label{prop_solv2}
Under the framework of section \ref{sec:system} (probability of an accumulation phenomenon described by a Generalized Pareto distribution) and the conditions of Corollary \ref{cor_extension}, condition (\ref{solvency}) holds for $\phi_{\beta}$ provided that $\theta\leq \eta(\alpha, \beta)\beta^{-1}$ and that, for some $\varepsilon'<\varepsilon$ and $a>1,$

\begin{equation}
\label{magic2}\eta(\alpha, \beta) \geq \max\left(\frac{\sigma \beta a S^{-1}(\varepsilon-\varepsilon')}{(a-1)N^{1/2} E[Y] \mu\left((0,\alpha_0+h_{\beta}(\tau)]\right)^{1/2}},\frac{a\beta s (1-\varepsilon'^{\gamma})}{\gamma\varepsilon'^{\gamma}}\right)
\end{equation}

which can also be rewritten as
$$\sup_{\mathbf{w}\in \mathcal{W}}\frac{m_Y(\alpha|\mathbf{w})-\phi_{\beta}(\mathbf{w})}{E[Y]}\leq 1-\beta+\theta_Y-\max\left(\frac{\sigma \beta a S^{-1}(\varepsilon-\varepsilon')}{(a-1)N^{1/2} E[Y] \mu\left((0,\alpha_0+h_{\beta}(\tau)]\right)^{1/2}},\frac{a\beta s (1-\varepsilon'^{\gamma})}{\gamma\varepsilon'^{\gamma}}\right).$$
\end{proposition}

The proofs of these two results are given in section \ref{proof_solv}.

Again, according to Propositions \ref{prop_solv1} and \ref{prop_solv2}, solvency can be achieved as long as $(m_{Y}(\alpha|\mathbf{w})-\phi_{\beta}(\mathbf{w}))/E[Y]$ is sufficiently small, which, as we mentioned earlier, can be interpreted as the ability of $\mathbf{W}$ to capture sufficient information on $Y.$ Let us note that, in this condition, the uniformity with respect to $\mathbf{w}$ is an important weak point: one expects to have situations where it is harder to approximate $Y$ from $\mathbf{W},$ leading to an increase in $Var(Y|\mathbf{W}=\mathbf{w}).$ This leads to the introduction, in the next section, of a hybrid product, mixing indemnity and index insurance, where the use of index insurance is restricted to the most favorable type of events.

\begin{remark}
Let us note that all the results we derive can be extended to other forms of payout functions. For a general payout $\phi$, the idea is to replace it with an ``equivalent'' contract in terms of utility, denoted $\phi_{\beta}$. Let $\beta_0$ be such that
$$E\left[\exp(\alpha(\phi(\mathbf{W})-Y))\right]=E\left[\exp(\alpha(\phi_{\beta_0}(\mathbf{W})-Y))\right].
$$
Then, both contracts are equivalent if they have the same price, that is, if $\pi_{\phi}=(1+\theta)\beta_0 E[Y]$, where $\theta$ is the loading factor of the product defined by $\phi_{\beta_0}$. Writing $\pi_{\phi}=(1+\theta_{\phi})\pi^*_{\phi}$, we obtain
$$
\theta_{\phi}=(1+\theta)\beta_0 E[Y]/\pi^*_{\phi}-1.
$$
Hence, all the results extend through this relationship between $\theta$ and $\theta_{\phi}$.
\end{remark}

\subsection{The hybrid product}
\label{sec:hybrid}

Let 
$$\mathcal{W}_{\alpha}(\frak{e},\beta)=\left\{\mathbf{w}\in \mathcal{W}:m_Y(\alpha|\mathbf{w})-\phi_{\beta}(\mathbf{w})\leq \frak{e}\right\}.$$
We define the following pay-off:
\begin{equation}
\label{hybrid}
\frak{h}^{\frak{e}}_{\alpha,\beta}(Y,\mathbf{W})= \exp(-\tau)Y\mathbf{1}_{\mathbf{W}\in \overline{\mathcal{W}_{\alpha,\beta}(\frak{e})}}+\phi_{\beta}(\mathbf{W})\mathbf{1}_{\mathbf{W}\in \mathcal{W}_{\alpha,\beta}(\frak{e})},
\end{equation}
where $\overline{\mathcal{A}}$ is the complementary set of the set $\mathcal{A}.$

The idea is that we use index insurance only in cases where we expect this solution to be reliable. From Proposition \ref{prop_existence}, we saw that the unfavorable situations for index insurance compensation are those where $m_{Y}(\alpha|\mathbf{w})-\phi_{\beta}(\mathbf{w})$ is large, which motivates the introduction of $\mathcal{W}_{\alpha,\beta}(\frak{e}).$ We consider that $m_{Y}(\alpha|\mathbf{w})-\phi_{\beta}(\mathbf{w})$ is large if it exceeds a certain threshold $\frak{e}$. In this case, policyholders will have a higher utility if they are compensated through indemnity-based insurance (see Section \ref{sec:exputil}). As with $\beta$, the choice of $\frak{e}$ is left to the discretion of the insurer. However, the latter should keep in mind that this choice will influence both the proportion of policyholders who prefer compensation through index insurance and the maximum loading factor $\theta^{\text{max}}$ at which index insurance can be sold to a given population. We believe that, in practice, the insurer will choose the value of $\frak{e}$ corresponding to this maximum loading factor, provided that the number of policyholders willing to accept compensation through index insurance is sufficient to make the product economically viable (see Sections \ref{sec:lln} and \ref{sec:system}). The impacts of the choice of $\frak{e}$ in a practical example are discussed in Section \ref{sec:discuss_e}.

The premium $\pi_{\frak{h}}$ associated with this product is
$$\pi_{\frak{h}}=(1+\theta_Y)E\left[Y|\overline{\mathcal{W}_{\alpha,\beta}(\frak{e})}\right](1-p_{\frak{e}}(\alpha,\beta))+(1+\theta)E\left[\phi_{\beta}(\mathbf{W})|\mathcal{W}_{\alpha,\beta}(\frak{e})\right]p_{\frak{e}}(\alpha,\beta),$$
where $p_{\frak{e}}(\alpha,\beta)=\mathbb{P}(\mathbf{W}\in \mathcal{W}_{\alpha,\beta}(\frak{e})).$ We apply the same loading factor $\theta_Y$ as for indemnity insurance to cases where exact compensation is offered, the lower loading factor $\theta$ being applied only to the index part. The higher loading factor of indemnity insurance directly reflects the increased claims management costs associated with this type of coverage. This hybrid design bears a conceptual resemblance to the blended insurance scheme proposed by \cite{zhang2024blended}.

Proposition \ref{prop_hybrid} provides a condition for the hybrid product $\frak{h}^{\frak{e}}_{\alpha,\beta}$ to be chosen instead of the traditional indemnity-based contract.

\begin{proposition}
\label{prop_hybrid} 
If
\begin{equation}
\label{eq:conditional}
\eta_{\frak{e}}(\alpha,\beta)= 1-\beta+\theta_Y-\frac{\frak{e}}{E\left[Y|\mathcal{W}_{\frak{e}}(\alpha,\beta)\right]p_{\frak{e}}(\alpha,\beta)},
\end{equation}
then, if $\theta\leq \eta_{\frak{e}}\beta^{-1},$ the policyholder with risk aversion less than $\alpha$ prefers the contract defined by the pay-off $\frak{h}^{\frak{e}}_{\alpha,\beta}$ for all $\tau\geq 0.$
\end{proposition}

From the fact that (\ref{eq:conditional}) should be non-negative, we see that the set $\mathcal{W}_{\alpha}(\frak{e},\beta)$ should not be too small, otherwise the probability $p_{\frak{e}}(\alpha,\beta)$ could make the left-hand side larger than $1-\beta+\theta_Y.$ On the other hand, one could be tempted to take a low value for $\frak{e}$ to control the difference between $m_{Y}(\alpha|\mathbf{w})$ and $\phi_{\beta}(\mathbf{w}),$ but this mechanically tends to make $\mathcal{W}_{\frak{e}}(\alpha,\beta)$ shrink. Let us also note that a too important increase of $\frak{e}$ introduces more constraints on $\theta:$ a decrease of $\eta_{\frak{e}}(\alpha,\beta)$ makes condition (\ref{eq:conditional}) easier to achieve, but the loading factor $\theta$ then should be smaller.

\begin{proof}
Similarly to the case of a full index product, the situation 
$\frak{U}_{\frak{h}^{\frak{e}}_{\alpha,\beta}}(\alpha)-\frak{U}_{Y,\tau}(\alpha)>0$ implies
\begin{equation}
\label{next}
E\left[\alpha\exp(Y-\frak{h}^{\frak{e}}_{\alpha,\beta}(Y,\mathbf{W}))\right]\leq E[\exp(\alpha' Y)]\exp(\alpha[\pi_Y-\pi_{\frak{h}}]).
\end{equation}
With $\alpha'=(1-\exp(-\tau))\alpha$

The left-hand side rewrites
\begin{align*}
    (1-p_{\frak{e}}(\alpha,\beta))E[\exp(\alpha' Y)] &+ E\left[\exp\left(\alpha\left\{m_Y(\alpha|\mathbf{W})-\phi_{\beta}(\mathbf{W})\right\}\right)\mathbf{1}_{\mathbf{W}\in \mathcal{W}_{\alpha,\beta}(\frak{e})}\right] \\ &\leq (1-p_{\frak{e}}(\alpha,\beta))E[\exp(\alpha' Y)] + \exp(\alpha \frak{e})p_{\frak{e}}(\alpha,\beta) \\ &\leq E[\exp(\alpha' Y)]\left\{(1-p_{\frak{e}}(\alpha,\beta)) + \frac{\exp(\alpha \frak{e})}{E[\exp(\alpha' Y)]}p_{\frak{e}}(\alpha,\beta)\right\} \\ &\leq E[\exp(\alpha' Y)]\exp(\alpha \frak{e}).
\end{align*}

Moreover,
$$\pi_Y-\pi_{\frak{h}}=(1+\theta_Y-\beta-\beta\theta)E[Y|\mathcal{W}_{\alpha,\beta}(\frak{e})]p_{\frak{e}}(\alpha,\beta).$$

Hence, a sufficient condition for (\ref{next}) is
$$\frac{\frak{e}}{E\left[Y|\mathcal{W}_{\alpha,\beta}(\frak{e})\right]p_{\frak{e}}(\alpha,\beta)}\leq 1-\beta+\theta_Y-\beta\theta,$$
which means that $\theta$ should be less than 
$\eta_{\frak{e}}(\alpha,\beta)\beta^{-1}.$ As for the proof of Proposition \ref{prop_existence}, $\frak{U}_{\frak{h}^{\frak{e}}_{\alpha,\beta}}(\alpha)-\frak{U}_{Y,\tau}(\alpha)>0$ implies that $\frak{U}_{\frak{h}^{\frak{e}}_{\alpha,\beta}}(\tilde{\alpha})-\frak{U}_{Y,\tau}(\tilde{\alpha})>0$ for any $\tilde{\alpha}\leq \alpha.$
\end{proof}

\subsection{Generalization to other forms of utility functions}

The exponential utility function is, of course, a specific choice that may not fully reflect the complexity of policyholders’ behavior. The choice of an exponential utility function is mainly driven by simplicity and the tractability of the resulting expressions.

The results can easily be extended to utility functions that are locally close to the exponential, under appropriate conditions.

For example, let us consider the CRRA utility function,
$$U(x)=\frac{x^{1-\gamma}}{1-\gamma}.$$
If $x=x_0+h$ remains close to a given value $x_0$, the function $U$ can classically be approximated as
$$
U(x)=\frac{x_0^{1-\gamma}\exp\left((1-\gamma)\frac{h}{x_0}\right)}{1-\gamma}+o\left(h\right).
$$

Hence, up to multiplication by a constant, the utility function is locally equivalent to an exponential utility with risk aversion parameter $\alpha=(1-\gamma)/x_0$. If $\phi(\mathbf{W})-Y$ remains small with high probability, the corresponding expected utility functions also remain close to each other. Consequently, all the results presented in the previous section can be generalized, with the lower bounds multiplied by a constant, and can therefore be interpreted as approximations.

\section{Practical illustration}
\label{sec:practical}

We illustrate in this section the results obtained previously with an example in the field of cyber insurance. We begin by studying the solvency requirements discussed in sections \ref{sec:lln}, \ref{sec:system} and \ref{sec:expsolvency}, and we end by showing how to deduce in which cases index insurance is relevant (see section \ref{sec:hybrid}). In section \ref{sec:context}, we describe the database and the choices made for the different parameters required in our models. Section \ref{sec:solvencyparam} focuses on the requirements of index insurance to ensure solvency, and finally, section \ref{sec:laplace} is devoted to the design of the hybrid product presented in section \ref{sec:hybrid}.

\subsection{Description of the context and the database}
\label{sec:context}

One of the complexities of cyber insurance lies in the difficulty of assessing the economic losses associated with a cyberattack. These losses are multifaceted and cannot be reduced to just damage to digital assets. Numerous high-profile examples, such as the Wannacry attack or the Colonial Pipeline breach, show that the damages are less related to the ransom amount than to the business interruption.

\textbf{The database.}

We consider a synthetic database of business interruptions related to cyber attacks\footnote{Available here: \url{https://github.com/dnkameni/cyber_cloud_interruption}}. The events contained in this database reflect real-life situations, but random noise has been introduced for the sake of confidentiality, and the data have been rescaled to reflect market trends which are consistent with the premium levels we consider in the following sections.

The database contains $n=10\,000$ claims. On each claim $i$, the following information is available:
\begin{itemize}
\item the loss amount $Y_i;$
\item the duration of the business interruption $T_i;$
\item $X_i\in \{t_1,\cdots,t_5\}$ is a qualitative variable related to the type of service that is impacted ;
\item an indicator variable $\delta_i \in \{0,1\}$, equal to 1 if the policyholder was able to activate a backup plan that potentially mitigated the impact of the business interruption;
\item an indication of the quality of the backup plan $B_i\in (0,1),$ corresponding to the expected proportion of activity that is preserved once the backup plan has been triggered ; this quantity is anticipated from security audit performed before the claim ;
\item $\Lambda_i=(T_i-U_i)_+$ where $U_i$ is the delay before triggering the backup plan. 
\end{itemize}
The available information just after the claim is $\mathbf{W}_i=(T_i,X_i,\delta_i,B_i,\Lambda_i).$ 

Some descriptive statistics are provided in Table \ref{tab:desc}. We see that in the absence of a backup plan, policyholders are on average exposed to a loss nearly 6\,000 euros higher than those with a backup plan. This difference is even more striking when considering that the higher losses in the absence of a backup occur over significantly shorter average interruption durations. For further analysis, see Figure \ref{fig:desities} in the appendix, which presents the densities of these two variables.

\begin{table}[h!]
    \centering
    \begin{tabular}{|c|c|c|c|c|}
    \hline
       Variable  & Mean & Minimum & Maximum & Standard deviation  \\
       \hline
       $Y$ (in $10^3$ euros)  & 109.08 & 1.09 & 2\,128.59 & 113.37 \\
       $T$ (in days) & 2.09 & 0.01 & 19.65 & 2.39 \\
       \hline
       $Y,\; \delta=1$ & 95.56 & 1.09 & 902.21 & 85.31 \\
       $T,\; \delta=1$ & 3.20 & 0.01 & 16.09 & 2.64   \\
       \hline
       $Y,\; \delta=0$ & 118.25 & 2.01 & 2\,128.59 & 128.16\\
       $T,\; \delta=0$ & 1.33 & 0.01 & 19.65 & 1.85 \\
        \hline
    \end{tabular}
    \caption{Descriptive statistics for the database.}
    \label{tab:desc}
\end{table}

The duration of business interruption $T$ is a key factor in assessing the consequences of cyber attacks, as highlighted in several studies, including \cite{tam2023quantifying}. Indeed, this variable is expected to be strongly correlated with the amount of loss. In our dataset, the linear correlation coefficient between losses and duration of business interruption is $0.57$. This value is $0.65$ for policyholders who activated a backup plan and $0.75$ for policyholders who failed to activate a backup plan. Figure \ref{fig:y_vs_t}, which shows a plot of losses due to business interruptions against the duration of the interruptions, appears to confirm these observations. Moreover, the duration of a business interruption can be measured immediately after the occurrence of the interruption, without requiring lengthy and costly expert evaluations. It is therefore quite natural to use this variable as the core component for constructing an index-based insurance product. Generally, the impact of a given duration of service interruption depends heavily on the sector of activity and the size of the affected company. In this simplified example, we consider a portfolio composed of policyholders with homogeneous profiles — that is, operating in the same sector, and of similar size.

\begin{figure}
    \centering
    \includegraphics[width=0.7\linewidth]{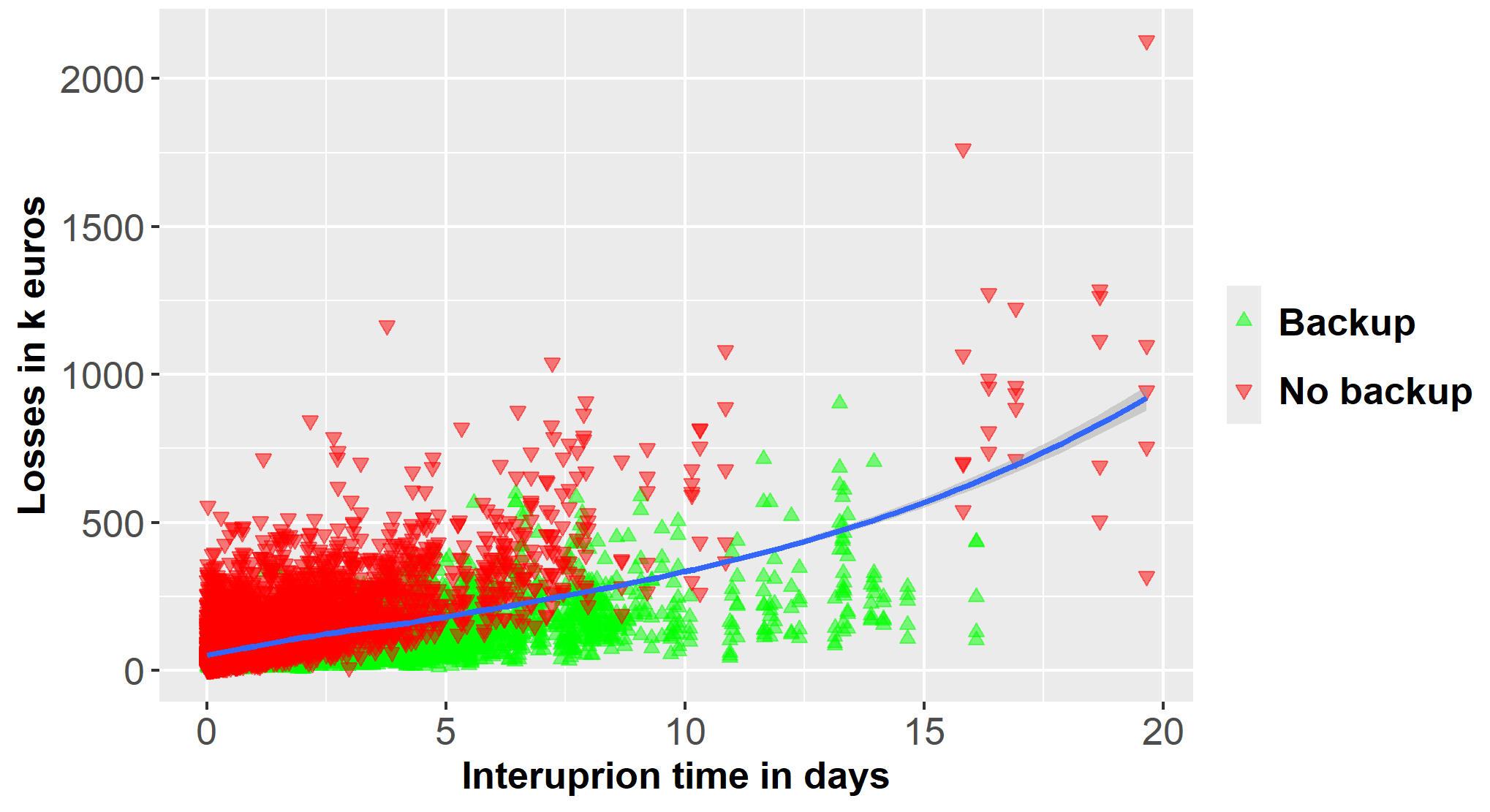}
    \caption{Losses $Y$ against duration of interruption $T$ with a distinction between policyholders whose backup plans were activated and those whose backup plans failed to kick-in. We observe an increasing relationship between losses and duration of interruption}
    \label{fig:y_vs_t}
\end{figure}

\textbf{Premium amount.}

In this simplified framework, we assume that each policyholder pays the same premium, calculated as the average compensation within the portfolio. To determine a reasonable value for the loading factor $\theta_Y$, we refer to the loss ratios observed on the market between 2020 and 2022 for medium-sized companies (according to the terminology used in the LUCY — Light Upon Cyber Insurance — report on the French market, published by AMRAE\footnote{Association pour le Management des Risques et l'Assurance de l'Entreprise. See the AMRAE report LUCY, \url{https://www.amrae.fr/bibliotheque-de-amrae/lucy-light-upon-cyber-insurance-2024-edition}}). The reported loss ratios are 45\% (2020), 36\% (2021), and 100\% (2022), yielding an average of approximately 60\%. For the purpose of this illustration we thus make the assumption that the premium $\pi_Y$ is calibrated to yield an average margin of 40\% above the pure premium, that is, we set $\theta_Y = 0.4$.

It is important to emphasize that this loading factor is not intended to accurately reflect actual market practices: premium levels can vary significantly from one customer to another, and loss ratios are not always anticipated. The current state of the cyber insurance market is still highly unstable. As shown in the same LUCY report, the amount of premiums has fluctuated considerably from year to year, illustrating the constant re-evaluation of risk by insurers in response to their prior underwriting results.

\textbf{Frequency of claims.}

The dataset we consider provides information about the severity of claims, but not about their frequency, as there is no indication of exposure. Thus, it informs us about the pay-off in the event of a claim, but not about the probability of such a claim occurring.

To assume a plausible value for frequency, we set the probability of experiencing an incident to $p = 0.06$. This value is inspired by the LUCY report (492 medium-sized companies in the sample, with 30 claims reported). Once again, we recall that, for simplification, our framework assumes that a policyholder does not experience more than one claim per year.

\textbf{Demand and risk aversion.}

In this application, we adopt an exponential utility function and must specify a measure $\mu$ to describe the distribution of risk aversion among potential policyholders. In the absence of a rigorous market study to assess price elasticity, modeling the distribution $\mu$ is the most questionable assumption. Again, the goal here is only to illustrate the methodology presented in this paper, not to provide a reliable estimate of cyber insurance demand.

We consider here that $\mu$ is a shifted exponential distribution, that is
$$d\mu(t)=\lambda \exp(-\lambda (t-\alpha_-))\mathbf{1}_{t\geq \alpha_-}.$$ To determine the value of $\alpha_-,$ we observe that the target population of policyholders who already subscribed an insurance contract accepted a price $\pi_Y.$ Since their preferences are described via an exponential utility, this is possible only if their risk aversion $\alpha$ is high enough. If we do not take into account the potential discount factor $\tau$ at this stage, this means that
$$\frac{\log \Psi_Y(\alpha_-)}{\alpha_-}=\pi_Y.$$ Estimating empirically $\alpha\rightarrow \Psi_Y(\alpha)$ from the database (see Figure \ref{fig_aversion}), we get $\alpha_-=0.049.$

Next, to consider a proper value for $\lambda,$ we make the assumption that half of the population of the policyholders is ready to accept an increase of $40\%$ of the premium. This choice is arbitrary, but is motivated by the fact the LUCY report noted an increase of 84\% of the collected premiums over 2022, for an increase of 53\% of the number of policyholders in the perimeter of the study (while deductible increase and insurance capacity stays stable). This seems to indicate that all of the current policyholders were ready to accept an increase of approximately 20\% of their premium. The proportion 40\% that is taken to set the value of $\lambda$ is then based on twice this number. This leads to $\lambda=45.08$. However, as illustrated in panel (a) of Figure \ref{fig:n_vs_a_tau_theta}, we also vary the value of $\lambda$ to observe the impact this has on the number of policyholders who accept index insurance, all things being equal.

\begin{figure}
    \centering
    \includegraphics[width=0.7\linewidth]{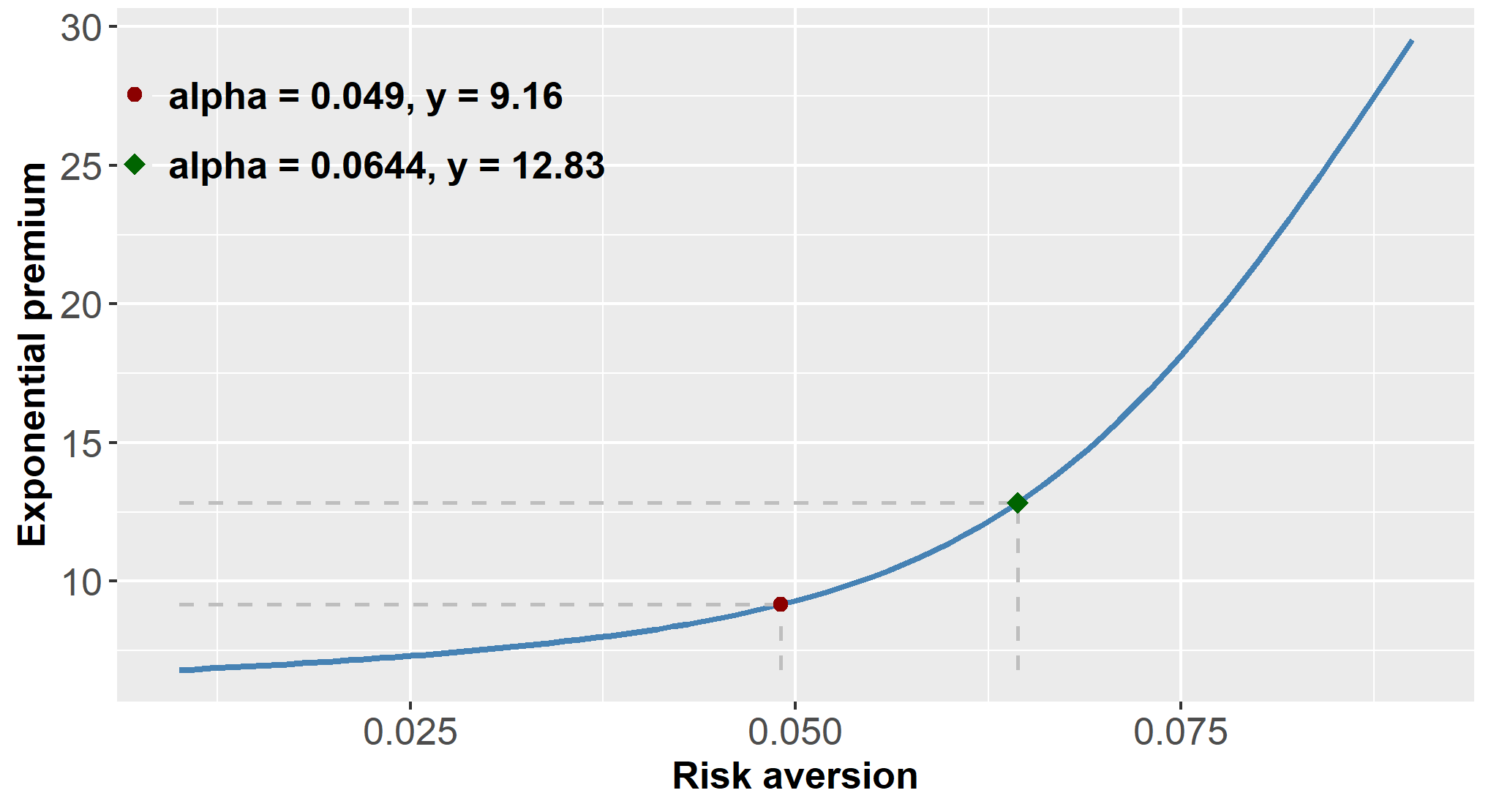}
    \caption{Estimation of $\alpha\rightarrow \log \Psi_Y(\alpha)/\alpha$ (exponential premium). The left point corresponds to the value of risk aversion corresponding to an exponential premium equal to $\pi_Y.$ The right point corresponds to the case where $\pi_Y$ increases by 40\%.}
    \label{fig_aversion}
\end{figure}

\subsection{Empirical analysis of the solvency requirements of index insurance}
\label{sec:solvencyparam}

In this section, we assume that a given population is exposed to an insurance market offering two types of insurance contracts: index insurance and indemnity-based insurance. We also assume that individuals are rational and will choose the insurance product that they expect to have the highest positive impact on their financial well-being. In other words, they will select the insurance product that maximizes their expected utility of wealth (see Section \ref{sec:demand}). We then study how the number of individuals in this target population who accept index insurance varies with certain factors that are likely to influence the demand for this type of insurance. This number is compared to the minimum number of policyholders needed to ensure, with high probability, the solvency of an index insurer in both classical situations and situations involving accumulation of losses (see Sections \ref{sec:lln} and \ref{sec:system}).

We set the tolerance level $\varepsilon$ at $0.5\%$. This choice is motivated by the regulatory requirement in Europe that insurers must hold sufficient capital to ensure a $99.5\%$ probability of remaining solvent over a one-year period (see \cite{scherer2021standard}). The distribution $\mu$ of the risk aversion coefficient $\alpha$ is set as described in section \ref{sec:context}. The values of the parameters used in the accumulation setting (see Section \ref{sec:system}) are set at $s = 0.003$, $\gamma = 0.5$, and $a = 2.4$. These values are chosen such that the condition $1 < a < \frac{\gamma \theta \varepsilon^{\gamma}}{s(1 - \varepsilon^{\gamma})}$ holds for $\theta > 0.18$ (recall that $\theta_Y = 0.4$).

To construct $\phi_\beta(\mathbf{W}) = \beta E[Y|\mathbf{W}]$, we set $\beta = 0.9$\footnote{A sensitivity analysis on $\beta$ is performed in Figure \ref{fig:theta_max}} and estimate $E[Y|\mathbf{W}]$ using five different models: a linear model (an approach similar to that of \cite{gine2007statistical}), a neural networks model, a regression tree model, a random forest model, and an eXtreme Gradient Boosting (XGBoost) model. The optimal hyperparameters obtained after tuning for these models are reported in Table \ref{tab:optimhyparam} in the appendix. The performances of these models on the in-sample data are presented in Table \ref{tab:models_phi}. The random forest and XGBoost models exhibit better performance than the linear and regression tree models. This indicates that these models have a higher probability of accurately predicting the losses suffered by policyholders for a given index value. 

\begin{table}[h!]
\centering
\begin{tabular}{lcccc}
\hline
\textbf{Model} & \textbf{$RMSE$} & \textbf{$R^2$} & \textbf{$MAE$} & \textbf{$Correlation$} \\ \hline
\textbf{Linear model}    & 4.307 & 0.599 & 2.813 & 0.774 \\
\textbf{Neural networks}  & 4.239 & 0.612 & 2.569 & 0.790 \\
\textbf{Regression trees}  & 4.197 & 0.619 & 2.787 & 0.787 \\
\textbf{Random forests}   & 3.018 & 0.811 & 1.897 & 0.901 \\
\textbf{XGBoost}         & 2.601 & 0.856 & 1.741 & 0.925 \\ \hline
\end{tabular}
\caption{Models used to build the index payout $\phi$.}
\label{tab:models_phi}
\end{table}

These high accuracies support the use of index insurance, as shown in Figure \ref{fig:n_vs_tau}. This figure presents a plot of the number of policyholders who accept index insurance in the target population against the delay in compensation of the competing indemnity-based insurance product. We observe that, for a given compensation delay in the indemnity-based product, index insurance products whose payouts are built using random forest or XGBoost models are likely to attract more policyholders. This is explained by the fact that the superior performance of these models reduces basis risk, which is one of the major challenges of index insurance. The reduced basis risk increases the satisfaction policyholders derive from index insurance, thereby enhancing the product’s attractiveness. This highlights the positive impact that machine learning can have on the design and commercialization of index insurance. However, the less performant linear and regression tree models have the advantage of being more interpretable and easier to explain to policyholders and regulators.

\begin{figure}
    \centering
    \includegraphics[width=1\linewidth]{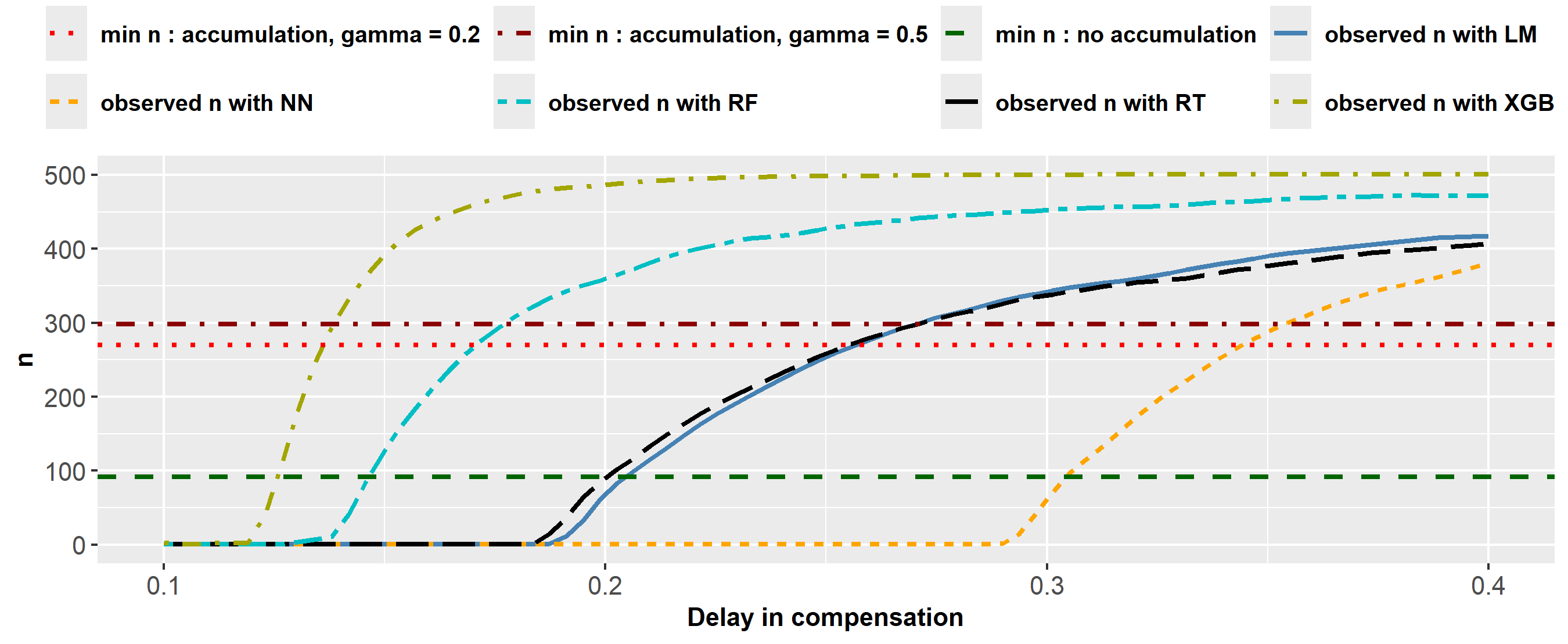}
    \caption{Number of policyholders $n$ in the population of size $N$ who prefer index insurance, as a function of the delay in compensation $\tau$ of the competing indemnity-based insurance product. The values of $n$ are represented for five index payout models namely a linear model (LM), a neural network model (NN), a random forest model (RF), a regression tree model (RT) and an extreme gradient boosting model (XGB). The reference values are computed for $\varepsilon = 0.5\%$}
    \label{fig:n_vs_tau}
\end{figure}

The results of the linear regression presented in Table \ref{tap:lm_regression} are consistent with what is typically expected in practice. Specifically, we observe that the losses suffered by policyholders increase with longer interruption times. Conversely, these losses tend to decrease in the presence of a backup plan, with longer durations of backup use, and with higher backup efficiency.

Constructing $\phi_\beta(\mathbf{W})$ is equivalent to learning the relationship between the loss $Y$ and the index $\mathbf{W}$. For this, we use the entire dataset described in Section \ref{sec:context}. However, we consider a target population of size $N = 500$, which approximately corresponds to the number of policyholders in the LUCY study (medium-sized companies). Figure \ref{fig:n_vs_a_tau_theta} shows how the number of policyholders $n$ in the population $N$ who accept index insurance varies with the mean risk aversion in the population $\Bar{\alpha}$, and the loading factor of index insurance $\theta$. This figure also illustrates the minimum number of policyholders required in the portfolio to ensure solvency with a $99.5\%$ probability under both classical and accumulation scenarios.

\begin{figure}[h!]
    \centering
    \begin{subfigure}{0.48\textwidth}
        \centering
        \includegraphics[width=\linewidth]{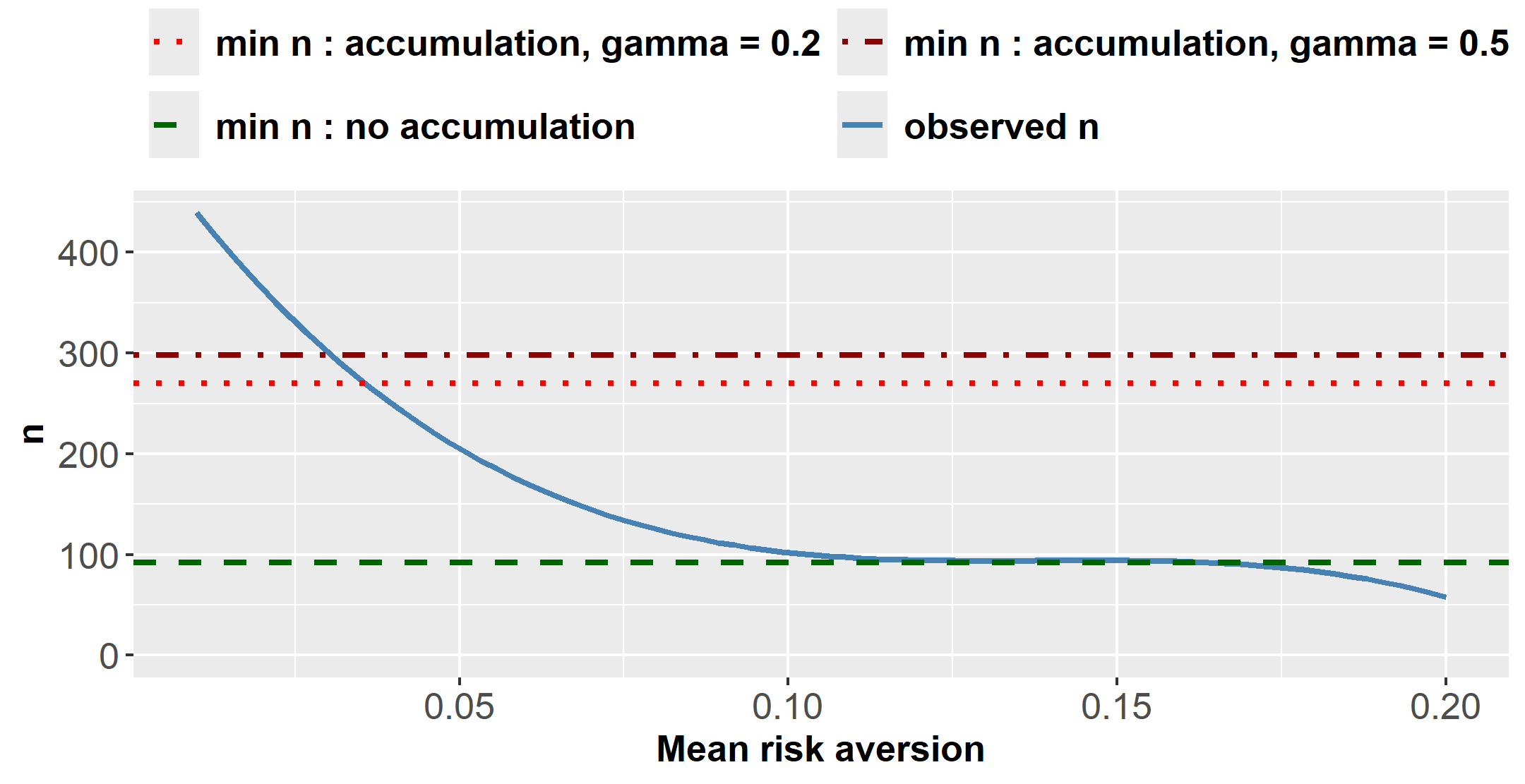}
        \caption{$n$ vs $\Bar{\alpha}$.}
    \end{subfigure}
    \begin{subfigure}{0.49\textwidth}
        \centering
        \includegraphics[width=\linewidth]{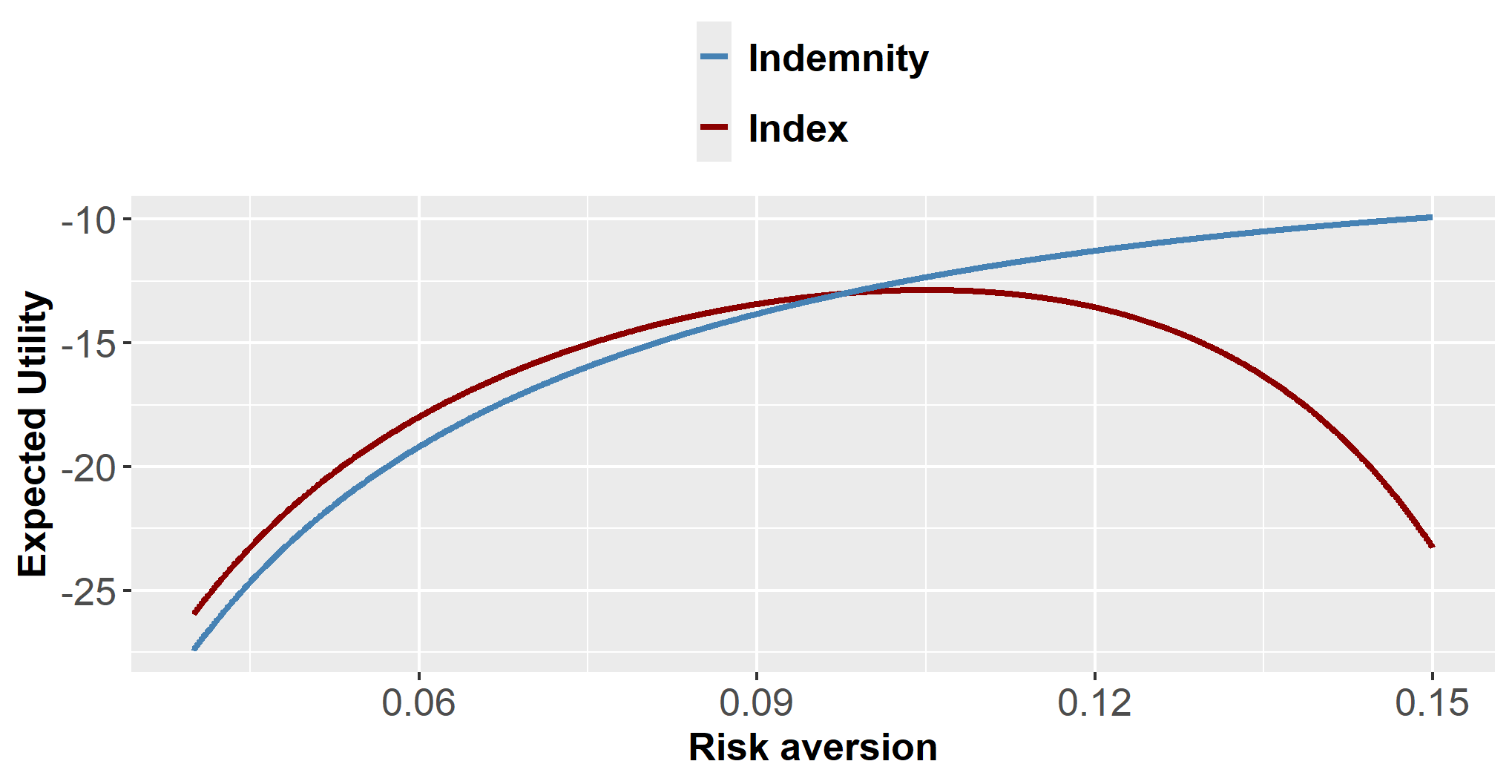}
        \caption{$\frak{U}$ vs $\alpha$.}
    \end{subfigure}
    \hfill
    \\
    \begin{subfigure}{0.7\textwidth}
        \centering
        \includegraphics[width=\linewidth]{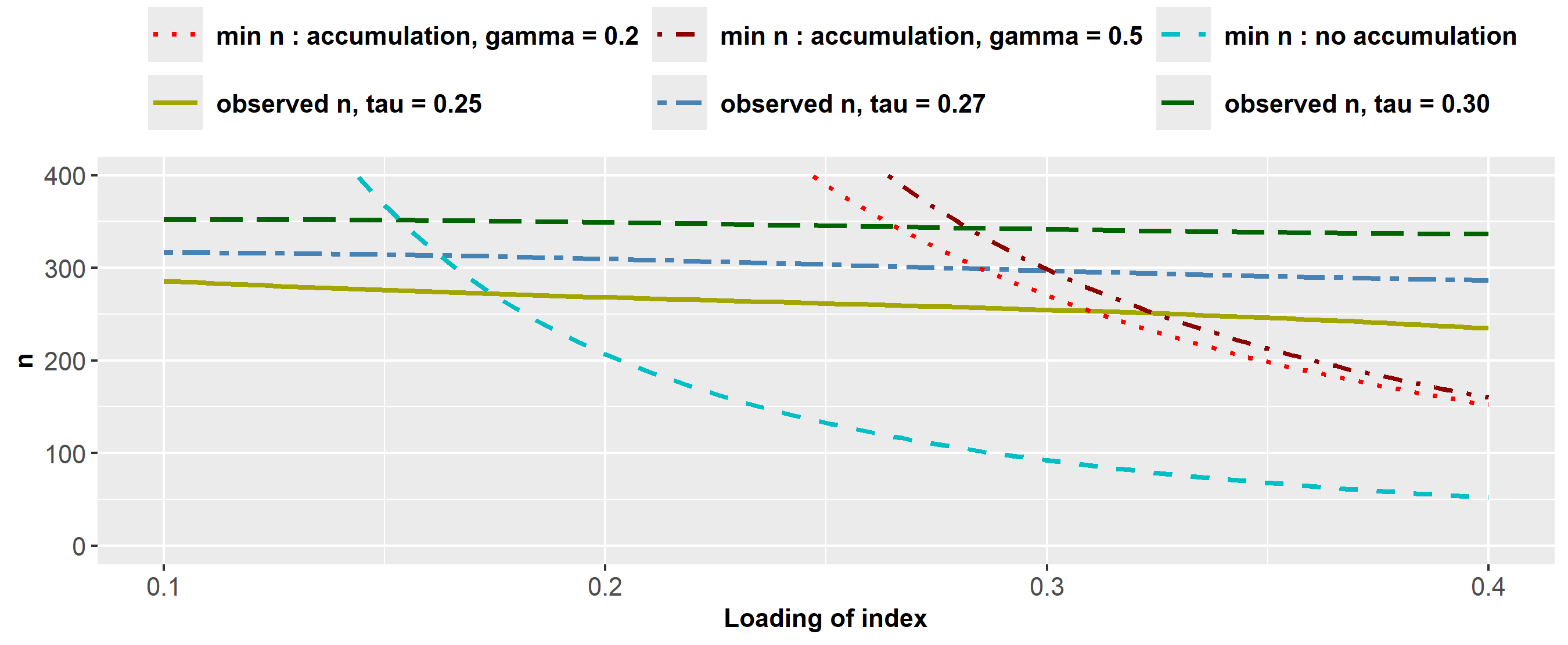}
        \caption{$n$ vs $\theta$.}
    \end{subfigure}
    \hfill
    \caption{Number of policyholders $n$ in the population of size $N$ who prefer index insurance, as a function of the mean risk aversion in the population $\Bar{\alpha} = \alpha_- + 1/\lambda$ (panel (a)) and the loading factor of index insurance $\theta$ (panel (c)). The reference values are computed for $\varepsilon = 0.5\%$. Panel (b) shows a plot of expected utility against risk aversion for index and indemnity-based insurance.}
    \label{fig:n_vs_a_tau_theta}
\end{figure}

As expected, the minimum number of policyholders in an index insurance portfolio required to ensure solvency with a probability of $99.5\%$ in the presence of accumulation is higher than the number needed when there is no accumulation. This is because greater mutualization is necessary to address accumulation and to efficiently ``dilute" large losses among policyholders. For the same reason, in the presence of accumulation, this minimum required number of policyholders increases with the heaviness of the tails of the distribution of accumulation events (see Section \ref{sec:system}), as shown in Figures \ref{fig:n_vs_tau} and \ref{fig:n_vs_a_tau_theta}.

Panels (a) and (b) of Figure \ref{fig:n_vs_a_tau_theta} empirically confirm the intuition that risk aversion works against index insurance. This conclusion aligns with that of \cite{clarke2016theory}, which states that the optimal demand for index insurance at any positive price is zero for infinitely risk-averse individuals. Panel (a) shows a steady decrease in the number of policyholders who accept index insurance as the mean risk aversion in the population increases, all else being equal. This results from the fact that, above a certain level of risk aversion, the utility policyholders derive from index insurance drops below that of indemnity insurance and continues to decrease rapidly with increasing risk aversion, as shown in panel (b). Let us note that, as risk aversion increases, both utility functions reach a maximum before decreasing (the maximum for the indemnity-based contract is attained outside the limits of the figure, but the curve exhibits the same shape). This behavior arises from the fact that all these expected utilities can be rewritten as $-\alpha^{-1}\exp(\alpha \pi)\,E[\exp(\alpha Z)]$, where $Z$ is a random variable that mostly takes positive values, and $\pi$ denotes the price of the contract under consideration. Examining the derivative of this expression reveals a change in monotonicity with respect to $\alpha$. When $Z$ tends to take larger values, the value of $\alpha$ at which this change in monotonicity occurs is reduced.

According to \cite{clarke2016theory}, the drop in utility of index insurance below that of indemnity insurance as risk aversion increases can be explained by the presence of basis risk\footnote{This is the risk that the payout of the index insurance contract differs from the actual loss suffered by policyholders \cite{clement2018global}} in index insurance. Indeed, basis risk can reduce the wealth of policyholders (on average), thereby lowering their expected well-being or expected utility. A risk-averse individual may therefore choose indemnity-based insurance instead, which in this setting provides higher compensation. However, moderately risk-averse individuals might behave differently, considering that the lower price of index insurance is sufficient to justify the basis risk. This distinct behavior of moderately risk-averse individuals is also highlighted by \cite{clarke2016theory}. Finally, note that this analysis holds if the delay in compensation of indemnity-based insurance is moderate. If the competing indemnity product on the market has a significantly longer compensation delay, then even extremely risk-averse individuals might opt for index insurance, as the delay in indemnity-based compensation would considerably reduce the expected present value of their wealth.

This observation appears to be confirmed by Figure \ref{fig:n_vs_tau}, which shows the relationship between the number of policyholders who prefer index insurance, $n$, and the delay in compensation of the competing indemnity-based insurance product, $\tau$, for a fixed distribution of risk aversion and a fixed index insurance loading factor. We observe that it might be impossible to launch an index insurance product on the market if it does not efficiently address the issue of delays in indemnity-based compensation. Specifically, if the difference in compensation speed between the two products is not sufficiently high in favor of index insurance, then the number of individuals in the population who accept index insurance will be insufficient to ensure solvency with a $99.5\%$ probability, regardless of whether accumulation is present. In practice, a low difference in compensation time could result from excessive paperwork or verification procedures imposed by a regulator on index insurance, to the extent that it loses its speed advantage.

The number of individuals in the population who prefer index insurance decreases with an increase in the loading factor $\theta$ of the index insurance contract, all else being equal. This is illustrated in Panel (c) of Figure \ref{fig:n_vs_a_tau_theta}, and it is a result that was expected. Indeed, as the loading factor of index insurance increases, policyholders are presented with a product that not only suffers from basis risk but is also priced closer to an equivalent indemnity-based insurance contract ($\theta_Y = 0.4$). Some individuals will then prefer to pay the same price for the indemnity-based product, which offers better coverage, assuming the compensation speed is reasonable. In practice, the lower loading factor of index insurance compared to its indemnity-based counterpart is achieved by eliminating claims management and expert assessment costs.

We also observe in panel (c) a decrease in the minimum number of policyholders required to ensure solvency at a $99.5\%$ probability as the index insurance loading factor increases. This is because, as index insurance becomes more expensive, the insurer earns more income from each policy sold. This increased income allows the insurer to cover all claims with a smaller portfolio, whether under classical loss or accumulation scenarios. This decreasing relationship between the solvency thresholds and $\theta$ enables the insurer to determine a minimum loading factor $\theta^{\text{min}} \approx 0.18$ at which the index insurance product can be sold competitively relative to indemnity-based insurance (with $\theta_Y = 0.4$). If an insurer were to offer the same index insurance product designed in this section to a population with the same characteristics as ours, applying a loading factor of $0.18$ would result in a portfolio of slightly more than 260 policyholders, which would be sufficient to ensure solvency with a $99.5\%$ probability in the absence of accumulation. Attempting to increase the number of policyholders by reducing the premium loading could expose the insurer to the risk of insolvency. An alternative solution could be to optimize claims management procedures (e.g., using Artificial Intelligence) to accelerate them, thereby increasing demand for the index insurance product as shown for different values of $\tau$ in panel (c). Another solution could be to educate the population about the advantages of index insurance, aiming to mitigate the influence of risk aversion on demand (see Panel (a)).

\subsection{Empirical analysis of a hybrid approach}
\label{sec:laplace}

In this section, we consider an insurance setting where an insurer can compensate losses using either an index insurance mechanism or an indemnity-based insurance mechanism. This insurer aims to propose each type of coverage where it is most suitable for policyholders and therefore more likely to be accepted. The identification of the various cases uses a criterion which is built using information on both the behavior and preferences of policyholders, as well as data on the index. However, the ultimate goal is to rely solely on index values to determine the type of coverage to apply. This section is an application of Section \ref{sec:hybrid}.

Identifying the cases where index insurance is more suitable for compensation amounts to determining the values of the index $\mathbf{w}$ who belong to the set
$$\mathcal{W}_{\alpha}(\frak{e},\beta)=\left\{\mathbf{w}\in \mathcal{W}:m_Y(\alpha|\mathbf{w})-\phi_{\beta}(\mathbf{w})\leq \frak{e}\right\},$$
for given values of $\frak{e}$ and $\beta$. The distribution $\mu$ of the risk aversion coefficient $\alpha$ is set as described in section \ref{sec:context}. This set corresponds to cases where the utility of wealth policyholders derive from index insurance compensation exceeds that which they derive from indemnity-based insurance compensation. To estimate the quantity $\Delta(\mathbf{w}) = m_Y(\alpha|\mathbf{w})-\phi_{\beta}(\mathbf{w})$—and thus determine whether $\mathbf{w}$ belongs to the set $\mathcal{W}_{\alpha}(\frak{e},\beta)$—we use two approaches:

\begin{itemize}
    \item regression tree models (see, for example, \cite{breiman2017classification} or \cite{loh2014fifty}) for their simplicity and interpretability;
    \item eXtreme Gradient Boosting (XGBoost) models for their good performances in predicting losses on our dataset (see section \ref{sec:solvencyparam}).
\end{itemize}

The choice of the model to use in practice will depend on the trade-off one is willing to make between model interpretability and model performance.

\subsubsection{\textbf{Discussion on $\frak{e}$ and methodology of analysis.}}
\label{sec:discuss_e}

Recall from Proposition \ref{prop_hybrid} that setting a value for $\frak{e}$ indirectly determines a maximum value $\theta^{\text{max}}=\eta_{\frak{e}}(\alpha, \beta)\beta^{-1}$ for the index insurance loading factor $\theta$. Additionally, the value of $\frak{e}$ influences the number or proportion ($p_{\frak{e}}(\alpha,\beta)$) of compensations for which index insurance is preferable, as it directly affects the size of the set $\mathcal{W}_{\alpha}(\frak{e},\beta)$. This dual influence of $\frak{e}$ is illustrated in Figure \ref{fig:theta_max}, using the same dataset described in Section \ref{sec:context}.

\begin{figure}[h!]
    \centering
    \begin{subfigure}{0.32\textwidth}
        \centering
        \includegraphics[width=\linewidth]{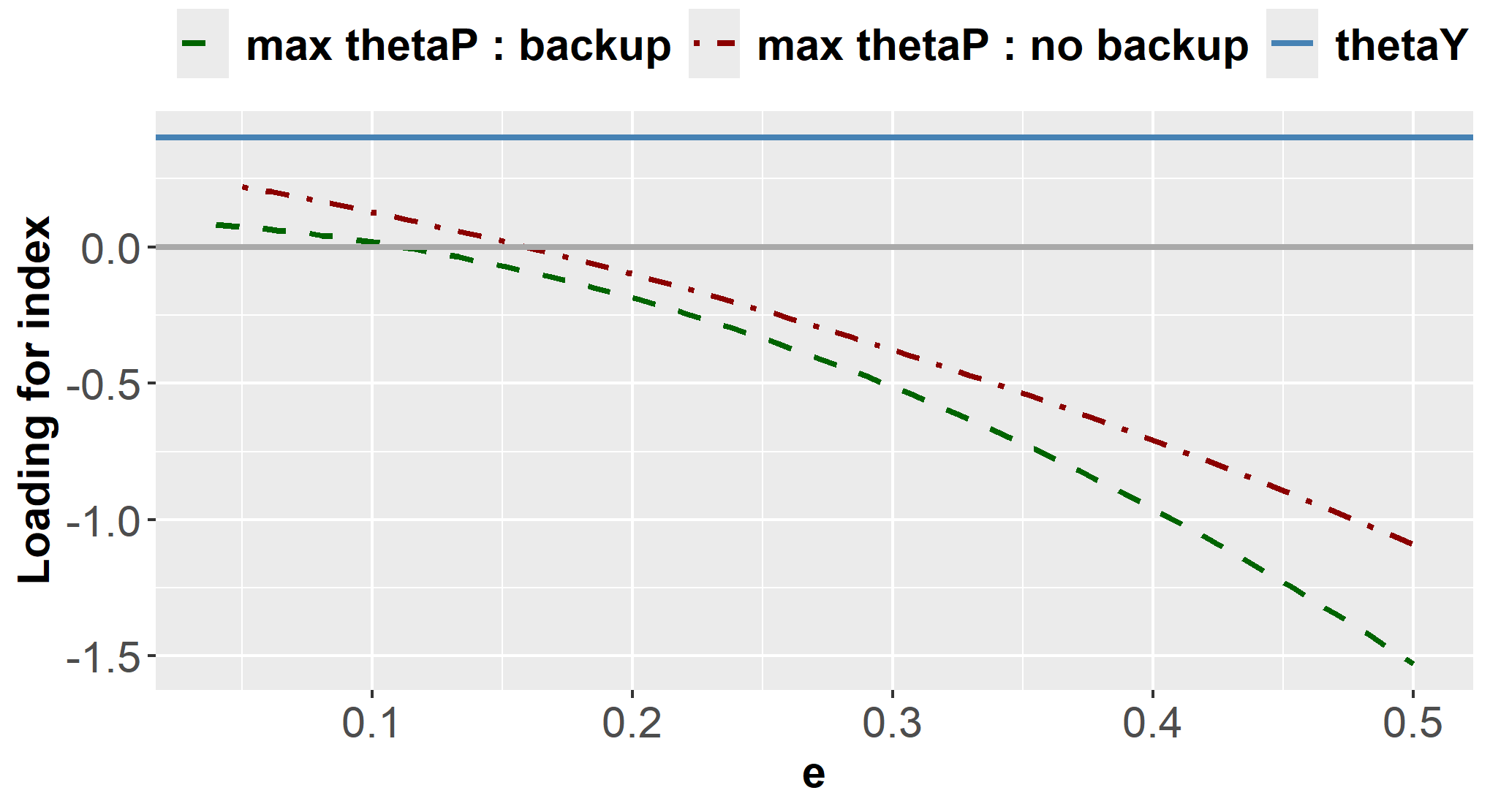}
        \caption{$\theta^{\text{max}}$ vs $\frak{e}$, $\beta = 0.80$.}
    \end{subfigure}
    \begin{subfigure}{0.32\textwidth}
        \centering
        \includegraphics[width=\linewidth]{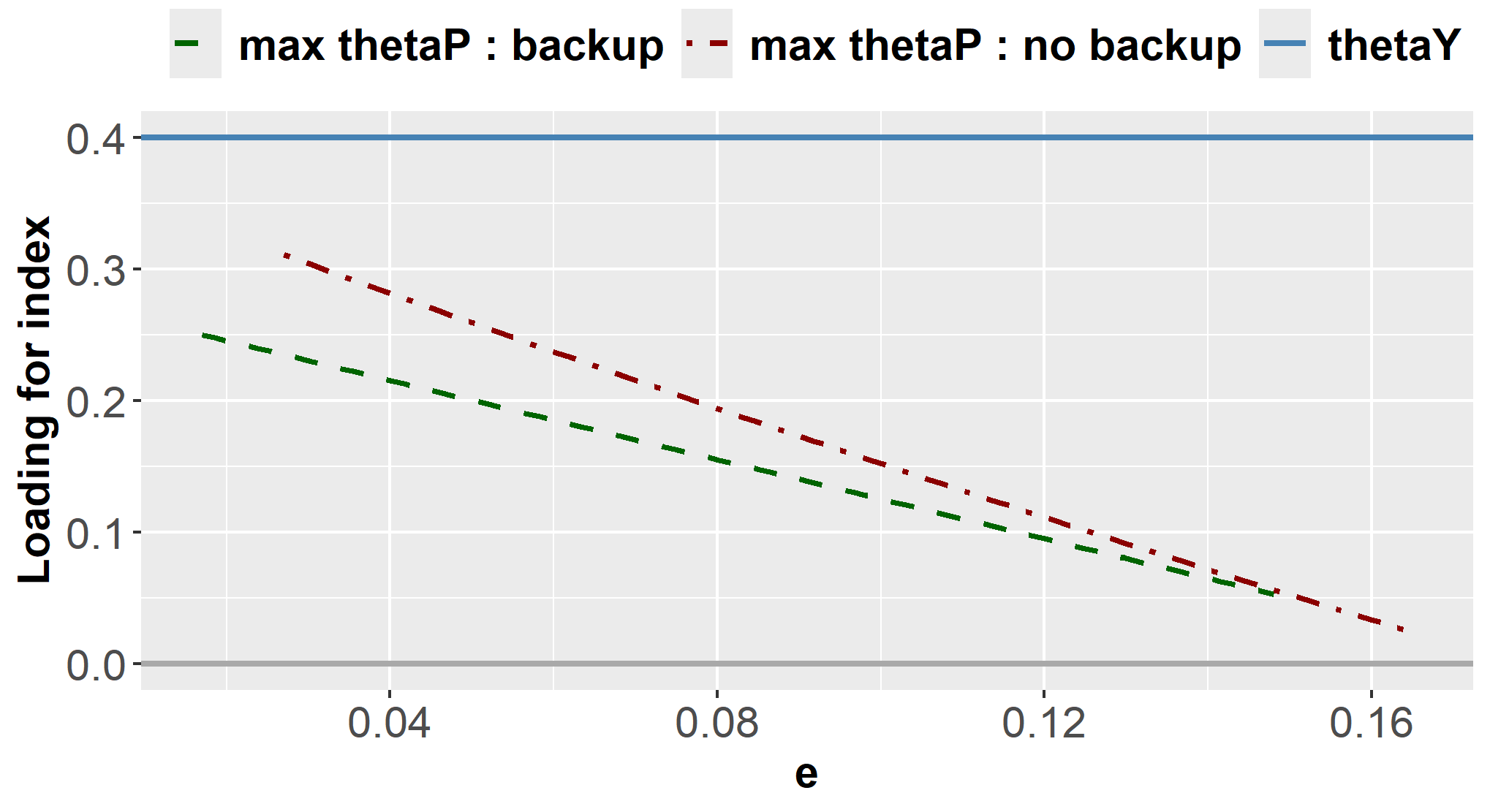}
        \caption{$\theta^{\text{max}}$ vs $\frak{e}$, $\beta = 0.90$.}
    \end{subfigure}
    \begin{subfigure}{0.32\textwidth}
        \centering
        \includegraphics[width=\linewidth]{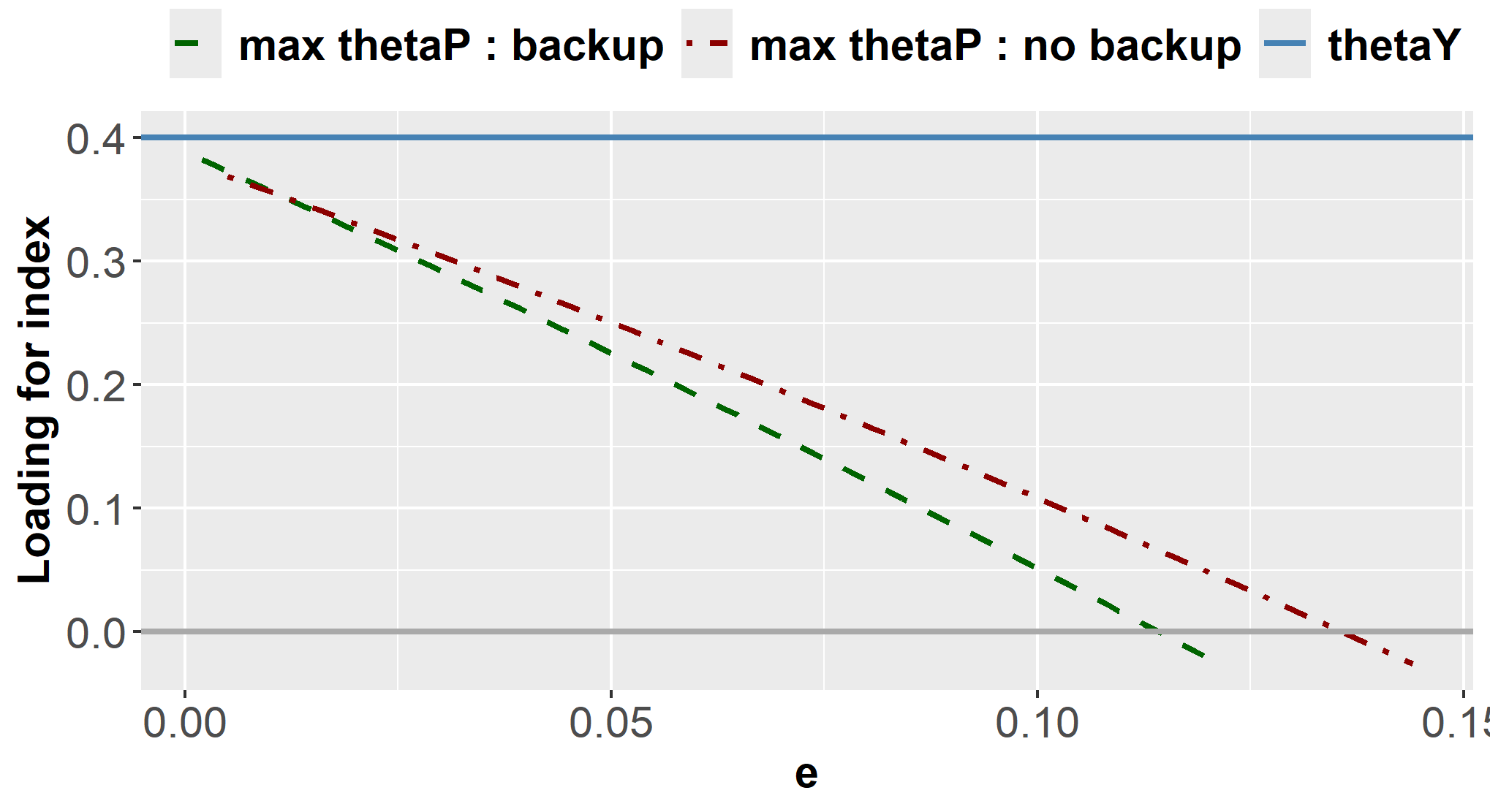}
        \caption{$\theta^{\text{max}}$ vs $\frak{e}$, $\beta = 1.00$.}
    \end{subfigure}
    \hfill
    \\
    \begin{subfigure}{0.32\textwidth}
        \centering
        \includegraphics[width=\linewidth]{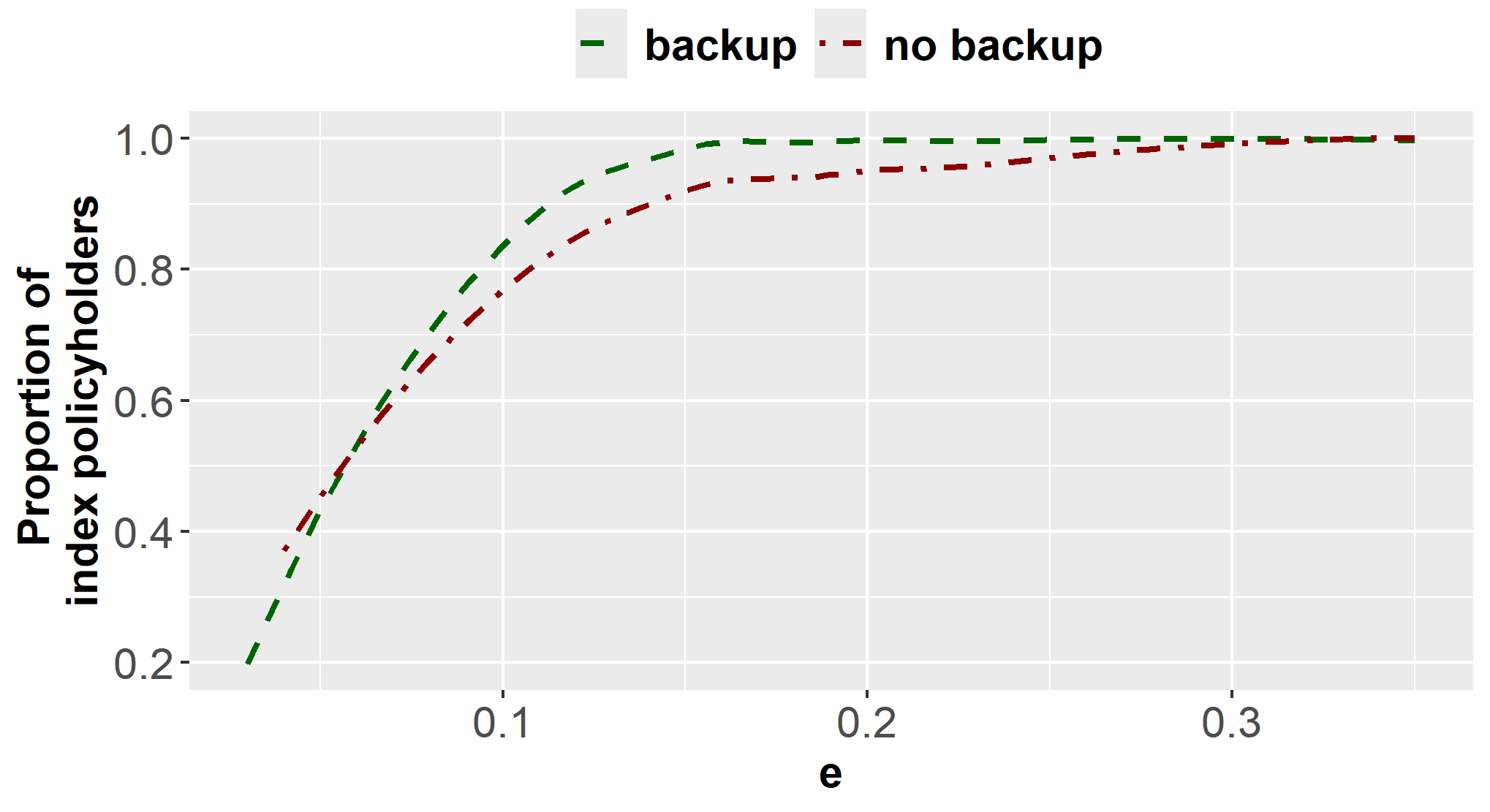}
        \caption{$p_{\frak{e}}(\alpha,\beta)$ vs $\frak{e}$, $\beta = 0.80$.}
    \end{subfigure}
    \begin{subfigure}{0.32\textwidth}
        \centering
        \includegraphics[width=\linewidth]{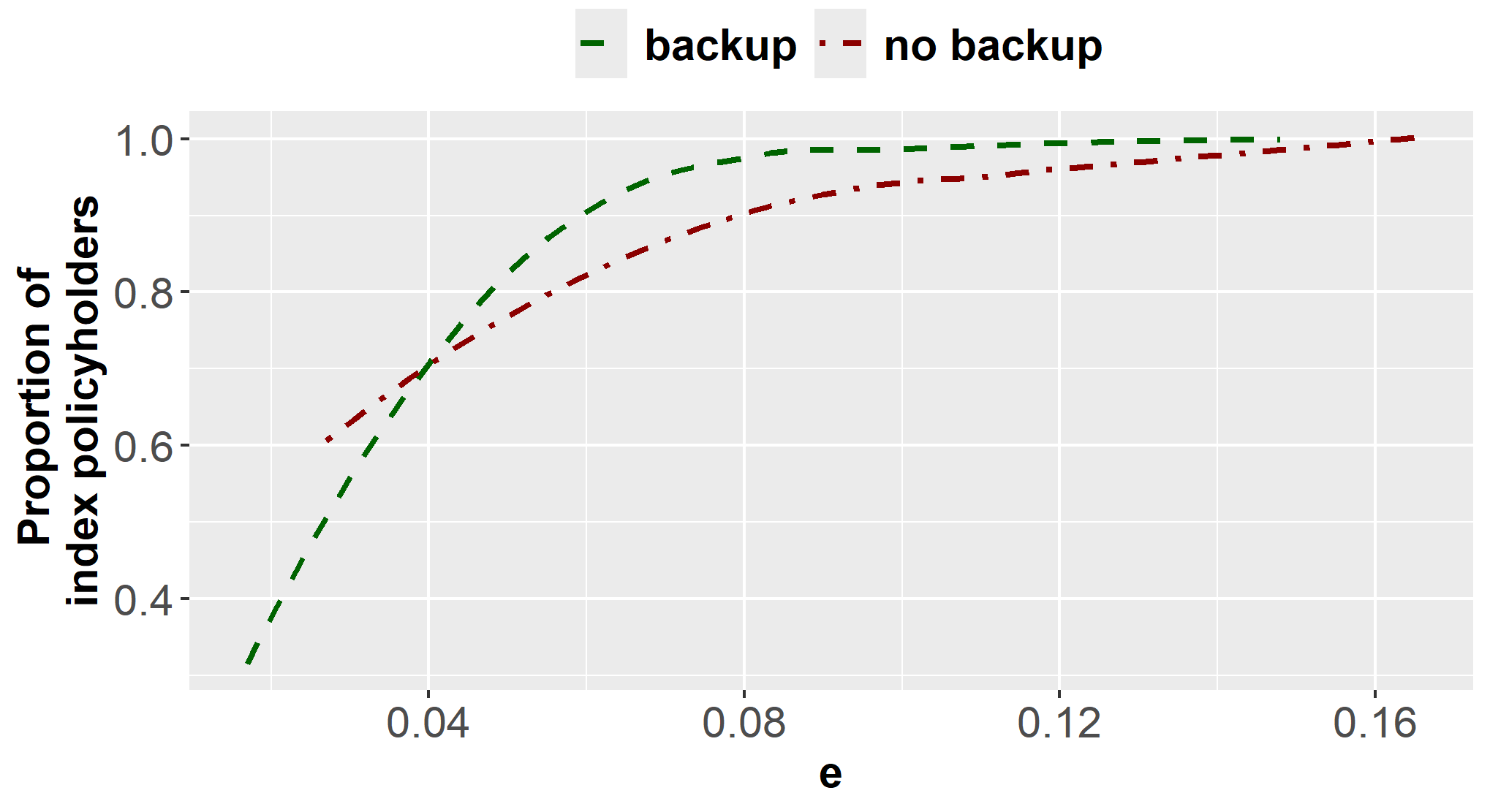}
        \caption{$p_{\frak{e}}(\alpha,\beta)$ vs $\frak{e}$, $\beta = 0.90$.}
    \end{subfigure}
    \begin{subfigure}{0.32\textwidth}
        \centering
        \includegraphics[width=\linewidth]{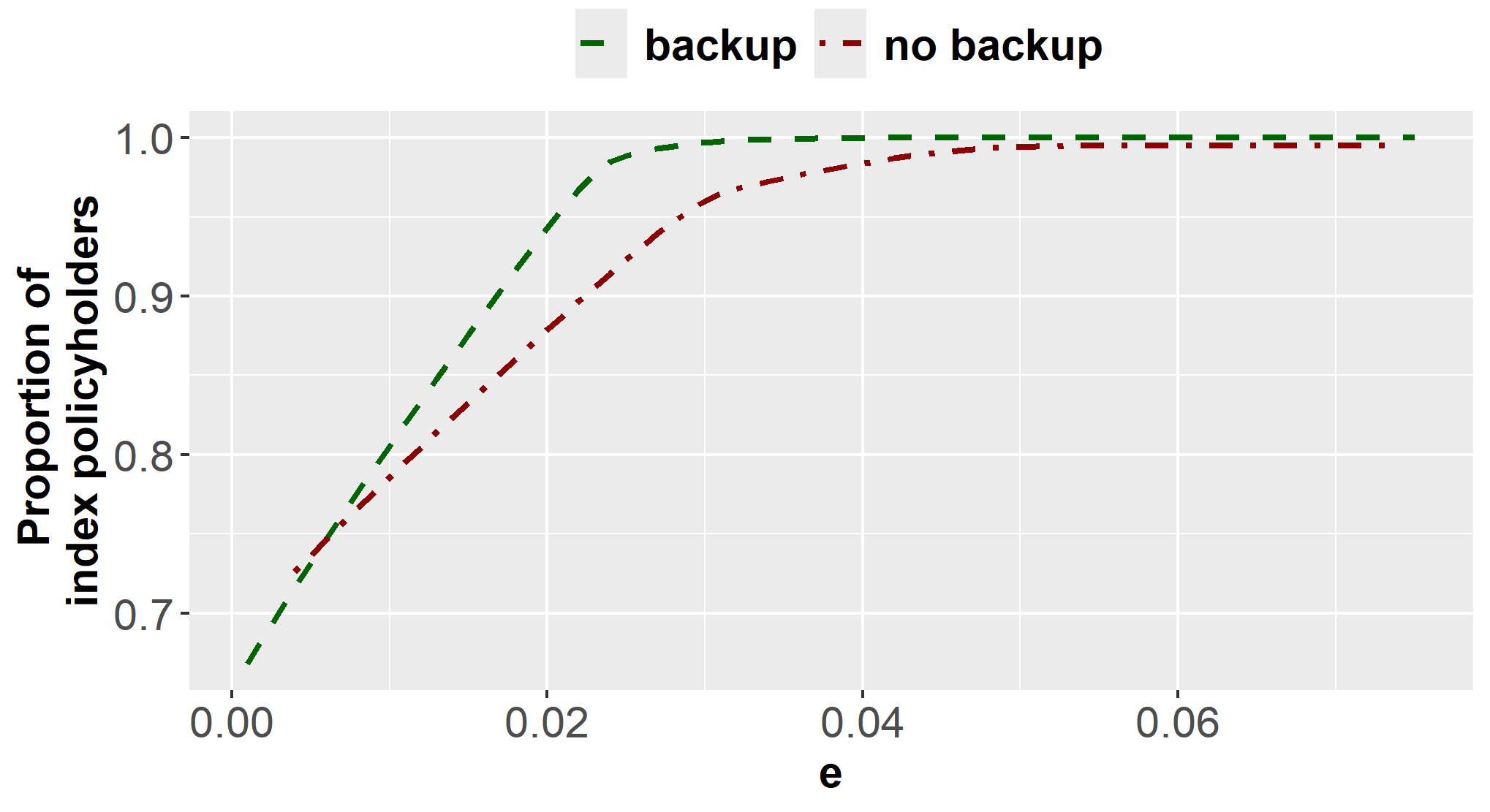}
        \caption{$p_{\frak{e}}(\alpha,\beta)$ vs $\frak{e}$, $\beta = 1.00$.}
    \end{subfigure}
    \hfill
    \caption{Maximum loading factor of index insurance $\theta^{\text{max}}= \eta_{\frak{e}}(\alpha, \beta)\beta^{-1}$ (panels (a), (b) and (c)) and proportion of compensations for which index insurance is preferable $p_{\frak{e}}(\alpha,\beta)$ (panels (d), (e) and (f)) as functions of the parameter $\frak{e}$ for various values of $\beta$. The value of $\theta_Y$ is also plotted in panels (a), (b), and (c) for comparison.}
    \label{fig:theta_max}
\end{figure}

Figure \ref{fig:theta_max} shows that increasing the value of $\frak{e}$ reduces the maximum loading factor $\theta^{\text{max}}$ of index insurance and thus lowers the price at which the index insurance contract can be sold to policyholders. Conversely, increasing $\frak{e}$ raises the proportion of compensations for which index insurance is more suitable. A direct explanation for this is that a higher value of $\frak{e}$ makes the condition $m_Y(\alpha|\mathbf{w}) - \phi_{\beta}(\mathbf{w}) \leq \frak{e}$ easier to satisfy, thereby increasing the size of the set $\mathcal{W}_{\alpha}(\frak{e},\beta)$. 

In practice, this means that when choosing $\frak{e}$, an insurer may need to balance between the proportion of index insurance compensations in their portfolio and the maximum loading that can be applied to the index insurance product. Indeed, Figure \ref{fig:theta_max} empirically demonstrates that increasing the share of index insurance compensation in a portfolio implies the necessity of offering the index insurance component of the hybrid contract at a lower price.

The final choice of $\frak{e}$ will be made following a solvency and ruin analysis, similar to that conducted in Section \ref{sec:solvencyparam}. This analysis will help the insurer select the value of $\frak{e}$ that maximizes revenue while minimizing the probability of ruin. The existence of such an optimal value is expected due to the dual and opposing effects of $\frak{e}$ on both the maximum loading of index insurance and the proportion of compensations for which index insurance is preferable. Other advantages of index insurance, such as faster compensation (which can foster client loyalty), reduced volatility, and simplified claims management, could also be factored into the insurer’s decision regarding the final proportion of index insurance compensations in their portfolio.

With this in mind, we propose Algorithm \ref{algo:index_select} to identify the cases where compensation via index insurance is the best option.

\begin{algorithm}
    \caption{Identification of compensation preferences}\label{algo:index_select}
    \begin{algorithmic}[1]
    \State $\alpha \gets$ estimate of a target risk aversion\footnotemark[1]
    \State $\beta \gets$ a suitable value to control overcompensation
    \State Set an optimal $\frak{e}^*$ based on $\theta^{\text{max}}(\alpha,\beta,\frak{e})$ and $p_{\frak{e}}(\alpha,\beta)$\footnotemark[2]
    \State Compute $\theta^{\text{max}}(\alpha,\beta,\frak{e}^*)$ and $p_{\frak{e}^*}(\alpha,\beta)$
    \State Train a model (regression trees and XGBoost in our case) to predict:
    $$\Delta(\mathbf{w}) = m_Y(\alpha|\mathbf{w}) - \phi_{\beta}(\mathbf{w})$$
    \State Flag all cases (leaves for regression trees and individuals for XGBoost) where $\displaystyle\sup_{\mathbf{w} \in \mathcal{W}} \Delta(\mathbf{w}) \leq \frak{e}^*$ as suitable for index insurance and the rest as suitable for indemnity-based insurance
    \end{algorithmic}
\end{algorithm}

\footnotetext[1]{According to proposition \ref{prop_hybrid}, if $\theta \leq \theta^{\text{max}}$, then all policyholders with risk aversion less than $\alpha$ will prefer the proposed hybrid contract compared to a full indemnity-based contract.}

\footnotetext[2]{The optimal value of $\frak{e}$ could be set by considering solvency requirements (see sections \ref{sec:lln} and \ref{sec:system}), profitability targets, and expert recommendations on the desired structure of the insurance portfolio (see figure \ref{fig:theta_max}).}

Applying this algorithm gives the insurer a tool to identify the cases where compensation via index insurance is more suitable. These cases constitute a proportion $p_{\frak{e}^*}(\alpha,\beta)$ of all compensations in the portfolio. Furthermore, the insurer has an indication on the maximum loading factor that can be applied to the index insurance component of the proposed hybrid product, given by $\theta^{\text{max}}(\alpha,\beta,\frak{e}^*)$. 

Regarding the choice of $\beta$, Figure \ref{fig:theta_max} shows that low values of $\beta$, although advantageous in considerably reducing the risk of overcompensation, also decrease the attractiveness of index insurance compensation. Indeed, panels (a) and (d) of Figure \ref{fig:theta_max} show that for $\beta = 0.8$, if the insurer applies the maximum possible loading to index insurance, only 20\% of the portfolio will prefer index insurance in the presence of a backup plan. This proportion rises to 70\% when $\beta = 1$ (see panels (c) and (f) of Figure \ref{fig:theta_max}). This can be explained by the fact that for very low values of $\beta$, the reduction in the price of the index insurance product fails to offset the decrease in utility or satisfaction experienced by policyholders due to the corresponding reduction in index insurance compensation. As a result, indemnity-based insurance compensation appears to be a more attractive option. In the rest of this section, we set $\beta = 0.9$. Since the efficiency of the backup plan plays a particular role in the severity of losses, we distinguish between claims associated with no backup plan (or a backup plan that was not triggered quickly enough, $\delta=0$) and those where the backup program successfully reduced the impact of the incident ($\delta=1$).

\subsubsection{\textbf{Applications with regression trees and XGBoost models.}}

In these applications, the values of $\frak{e}^*$ are chosen arbitrarily to illustrate our methodology and the proposed hybrid product. Figure \ref{fig:reg_tree} and figure \ref{fig:reg_xgb} display the two regression trees and the two scatter plots obtained in our setting for the regression trees and the XGBoost models respectively.

\begin{figure}[h!]
    \centering
    \begin{subfigure}{1.0\textwidth}
        \centering
        \includegraphics[width=\linewidth]{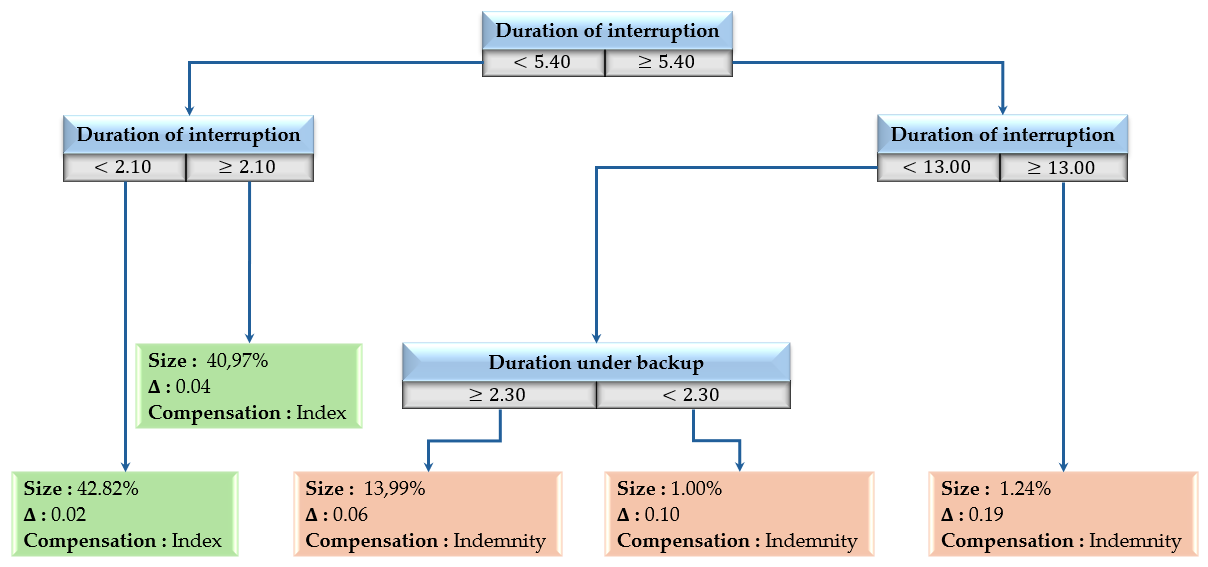}
        \caption{Backup plan was successful ($\delta=1$).}
    \end{subfigure}
    \\
    \begin{subfigure}{1.0\textwidth}
        \centering
        \includegraphics[width=\linewidth]{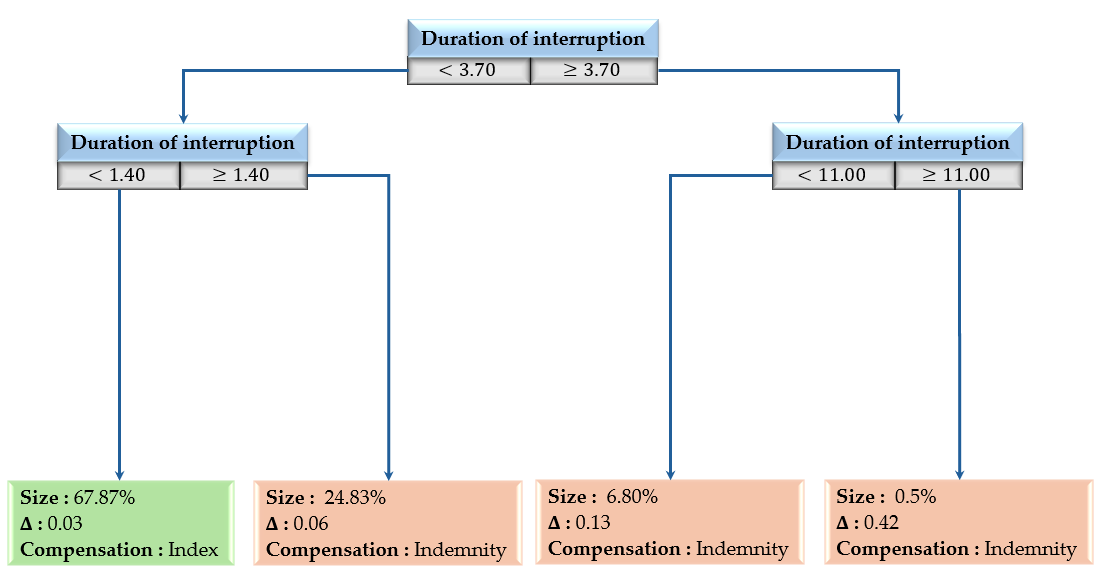}
        \caption{Backup plan failed ($\delta=0$).}
    \end{subfigure}
    \caption{Regression trees used to cluster compensations according to preferences in compensation type (index or indemnity). Panel (a) is the tree for situations in which the backup plan was successful and panel (b) is the tree for situations in which the backup plan failed}. For each leaf, information is provided on the proportion of compensations in that leaf (size), the value of $\Delta(\mathbf{w})$ for that leaf, and the decision (for the value of $\frak{e}^*$ chosen).
    \label{fig:reg_tree}
\end{figure}

Applying algorithm \ref{algo:index_select} with regression trees yields a maximum loading of the index insurance part of the contract ($\theta^{\text{max}}$) equal to $0.32$ in cases where the backup plan is triggered on time, and $0.26$ in cases where the backup plan fails or is absent. These values are respectively $0.27$ and $0.22$ when algorithm \ref{algo:index_select} is used with an XGBoost model. Recall that the loading factor for the indemnity-based part of the contract is $0.40$. When regression trees are used, our methodology reveals that compensation with index insurance is preferable in $83.79\%$ of cases, when the backup plan is triggered fast enough. This proportion is $67.87\%$ in situations where the backup plan fails or is absent. When an XGBoost model is used, the previous proportions are respectively $89.43\%$ and $49.60\%$. These percentages correspond to the shares of policyholders who suffered a loss in the portfolio who are likely to prefer compensation via index insurance.

Note that the XGBoost model allows for individual identification of the preferred or ideal compensation types, unlike the regression tree models, which perform a clustering of compensations and determine the ideal compensation type per cluster. This difference in approach, along with the difference in predictive performance between the two models, are among the reasons that could justify the differences in the values of $\theta^{\text{max}}$ and $p_{\frak{e}^*}(\alpha,\beta)$ obtained from the two models.

\begin{figure}[h!]
    \centering
    \begin{subfigure}{0.49\textwidth}
        \centering
        \includegraphics[width=\linewidth]{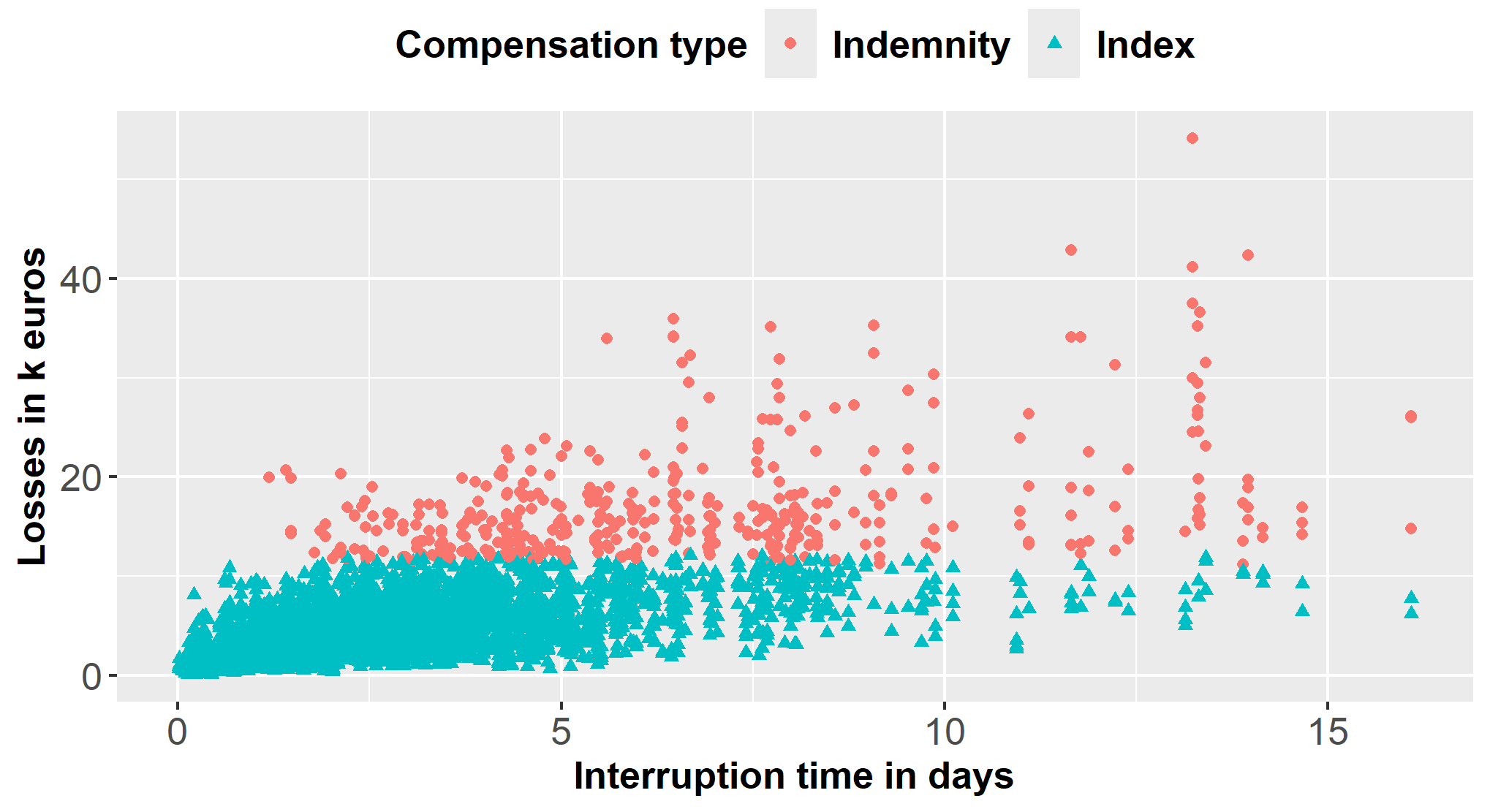}
        \caption{Backup plan was successful ($\delta=1$).}
    \end{subfigure}
    \begin{subfigure}{0.49\textwidth}
        \centering
        \includegraphics[width=\linewidth]{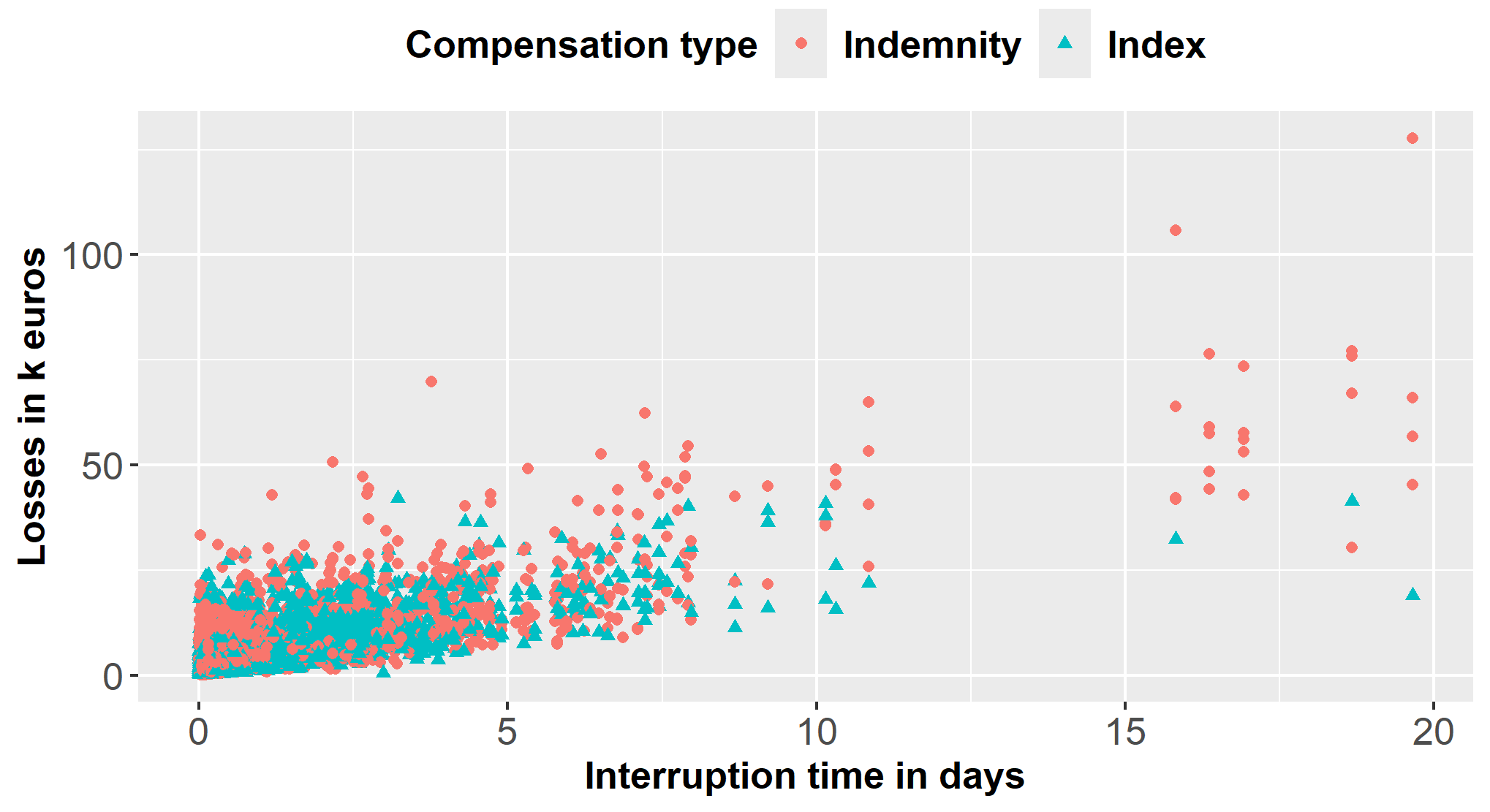}
        \caption{Backup plan failed ($\delta=0$).}
    \end{subfigure}
    \caption{Scatter plots of losses against interruption time. Panel (a) shows the plot for situations in which the backup plan was successful and panel (b) shows the plot for situations in which the backup plan failed. The ideal compensation type obtained by applying algorithm \ref{algo:index_select} with XGBoost models can be observed in both cases. In this case, $\Delta(\mathbf{w})$ is calculated at an individual level and not in clusters as is the case with regression trees models.}
    \label{fig:reg_xgb}
\end{figure}

The results presented in figure \ref{fig:reg_tree} and in figure \ref{fig:reg_xgb} support the use of index insurance compensation for short durations of business interruption and, consequently, for small losses (as clearly shown in panel (a) of figure \ref{fig:reg_xgb}). This is a favorable outcome, as index insurance tends to exhibit lower basis risk for small losses. Indeed, \cite{lopez2023parametric} show that basis risk generally increases with the size of the loss. In other words, our methodology recommends applying index insurance in contexts where it performs best. Combined with the other advantages of index insurance, this may explain why policyholders in such situations show a preference for it. 

Alternative segmentation techniques could also be used for this identification task. Artificial Intelligence could be incorporated to either improve the efficiency of the choice of $\frak{e}^*$ or enhance the quality of the segmentation of compensations. Once the cases in which a preference for compensation with index insurance are identified, a more detailed analysis of the claims and policyholders’ profiles within each preference group could be conducted to build an index insurance payout model tailored to each group where it is preferred. This approach could help reduce overall basis risk and improve the quality of the proposed hybrid insurance contract. Authors such as \cite{hernandez2023role}, \cite{lin2023evolution}, and \cite{cesarini2021potential} also highlight the significant contributions that Artificial Intelligence and Machine Learning could bring to the development of index insurance payout models.

\section{Conclusion}

In this paper, we proposed a framework to analyze the introduction of an index insurance product in competition with a traditional indemnity insurance product. The index insurance product should be attractive enough to convince a sufficient number of policyholders in order to ensure the solvency of the portfolio. Our analysis suggests that the number of policyholders willing to accept index insurance can be influenced by factors such as the risk aversion of policyholders, the delay in compensation of the competing indemnity-based insurance product, and the price of the index insurance product. 

We also propose the use of a hybrid product, where index insurance is used as compensation only in situations where it is most suitable for policyholders, with indemnity-based insurance used otherwise. This combination has several advantages, including accelerating compensation and reducing the premiums paid by policyholders. 

Additionally, we developed an algorithm to help insurers identify the specific losses for which a compensations using index insurance are more likely to be preferred and accepted by policyholders and those for which compensations using indemnity-based insurance are more suitable. Beyond the concept of index-based insurance, the results from the hybrid insurance section can also be applied for claim management in traditional indemnity-based insurance contracts. This could be relevant in situations where the insurer aims to reduce expert costs when handling certain claims or when an insurer wishes to apply different claims management policies to various segments of their portfolio based on specific criteria.

It is important to note that we do not cover one of the most appealing aspects of index insurance in this paper, which is offering coverage for claims not covered by traditional indemnity contracts. The present paper focuses solely on the introduction of index insurance in a context where traditional indemnity-based insurance is already in place. The question of modeling the demand for index insurance in situations where indemnity-based insurance is unavailable will be addressed in future work and will require a more nuanced approach to calibrate the utility function.

\textbf{Acknowledgment:} Olivier Lopez acknowledges funding from the Excellence Chair CARE (Allianz, Ensae, Risk Fundation).

\section{Appendix}

\subsection{Additional information}

\begin{figure}[H]
    \centering
    \begin{subfigure}{0.48\textwidth}
        \centering
        \includegraphics[width=\linewidth]{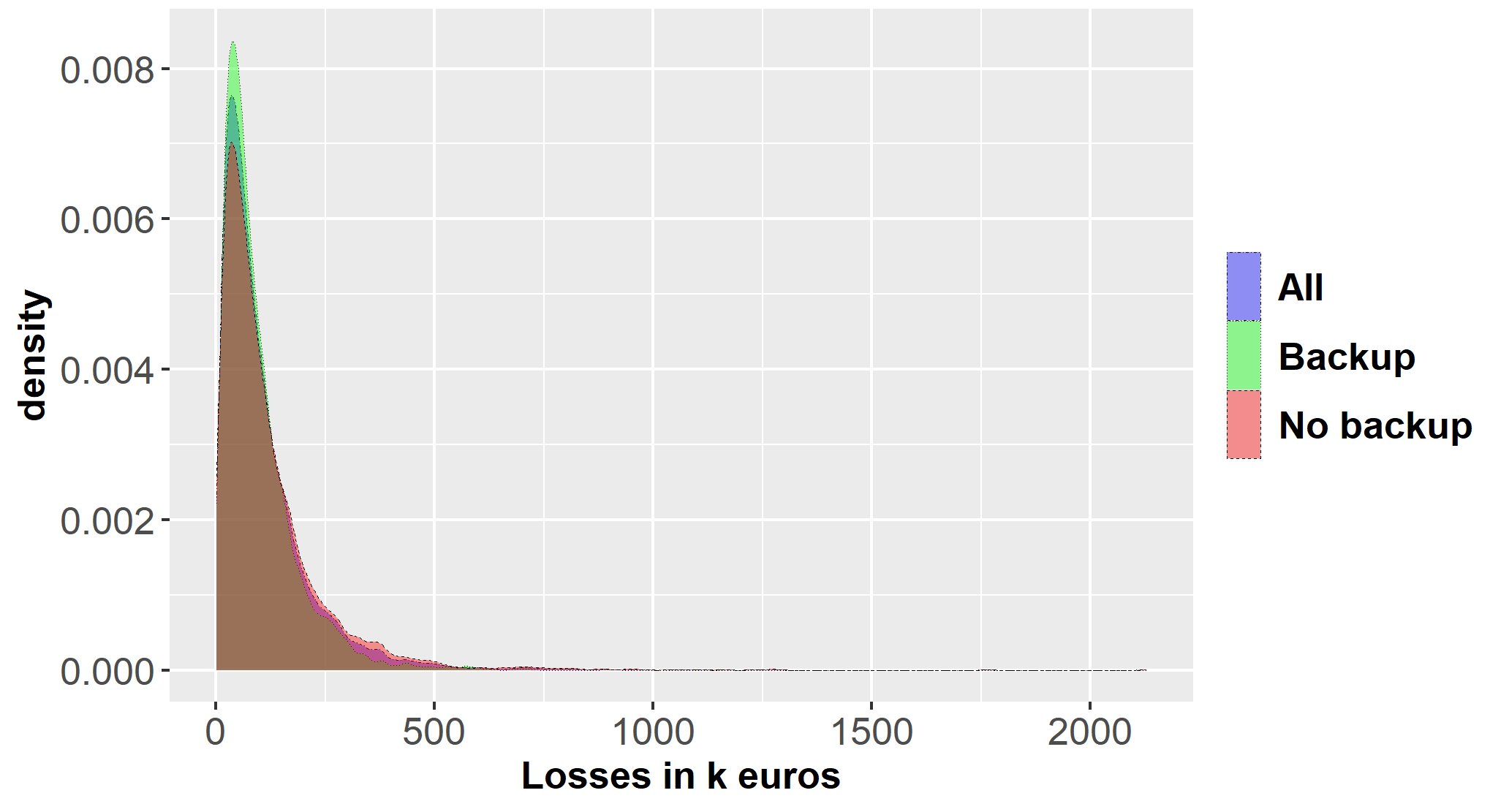}
        \caption{Financial losses in k euros.}
    \end{subfigure}
    \hfill
    \begin{subfigure}{0.48\textwidth}
        \centering
        \includegraphics[width=\linewidth]{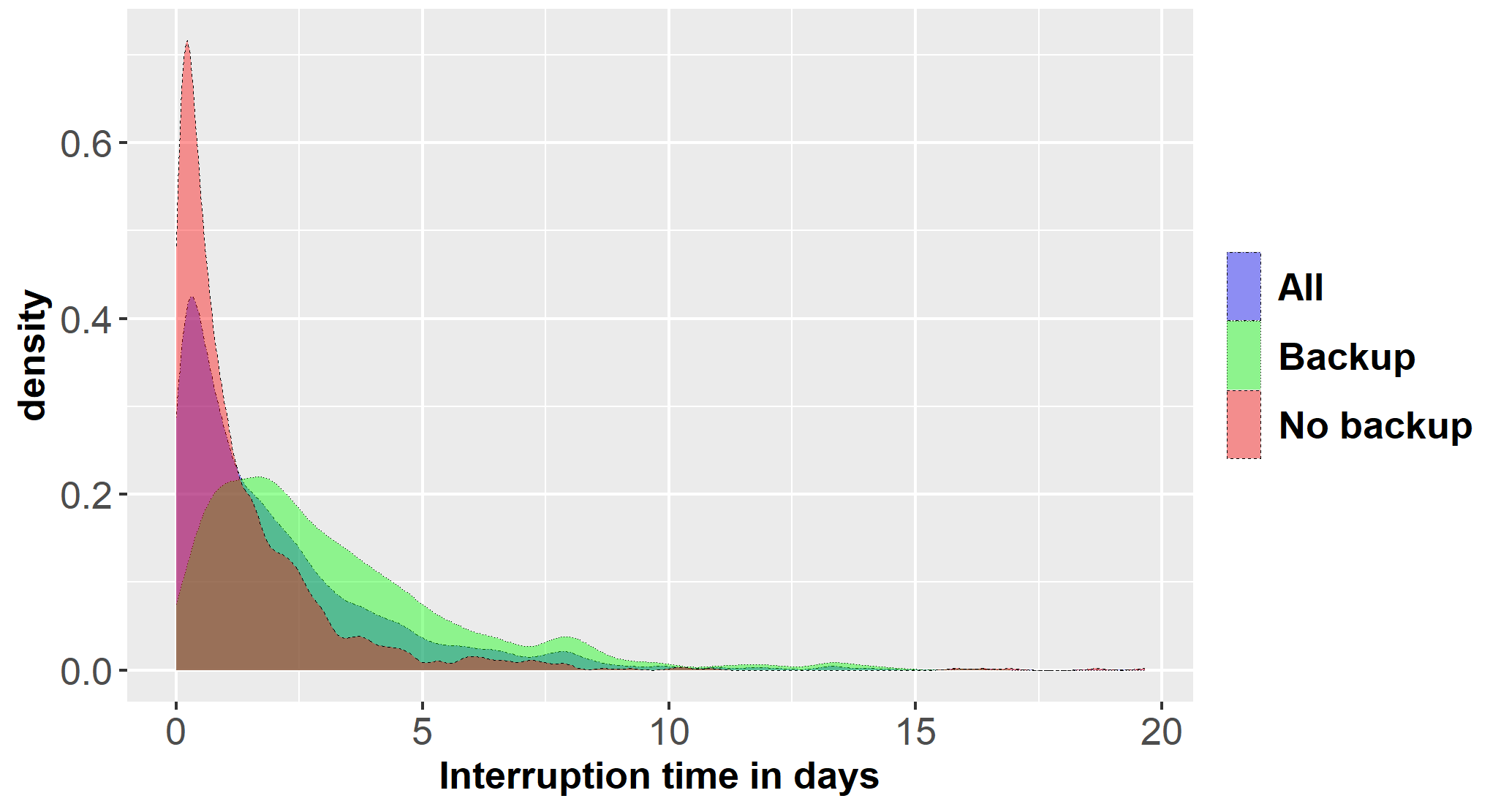}
        \caption{Duration of interruption in days.}
    \end{subfigure}
    \caption{Densities of losses due to cloud service interruptions (panel (a)) and densities of durations of interruption (panel (b)). The densities are plotted for all policyholders and separately for policyholders with backup plans and those without backup plans.}
    \label{fig:desities}
\end{figure}

\begin{table}[h!]
\centering
\begin{tabular}{lcccc}
\hline
\textbf{Coefficients} & \textbf{Estimate} & \textbf{Standard error} & \textbf{t value} & \textbf{p-value} \\
\hline
(Intercept)   & 4.30287  & 0.25261 & 17.033   & $< 2 \times 10^{-16}$*** \\
$T$   & 3.10957  & 0.02757 & 112.778  & $< 2 \times 10^{-16}$*** \\
$\delta$      & -2.16544 & 0.11037 & -19.620  & $< 2 \times 10^{-16}$*** \\
$\Lambda$           & -2.17856 & 0.03986 & -54.654 & $< 2 \times 10^{-16}$*** \\
$B$   & -3.53310 & 0.32735 & -10.793  & $< 2 \times 10^{-16}$*** \\
\hline
\multicolumn{5}{c}{\textit{\textbf{Significance codes:}} 0 ‘***’, 0.001 ‘**’, 0.01 ‘*’, 0.05 ‘.’, 0.1 ‘ ’} \\
\hline
\multicolumn{5}{c}{\textit{\textbf{Multiple R-squared:}} 0.5990 \quad \textit{\textbf{Adjusted R-squared:}} 0.5987} \\
\hline
\end{tabular}
\caption{Results of the linear regression of $Y$ on $\mathbf{W}$}
\label{tap:lm_regression}
\end{table}

\begin{table}[h!]
\centering
\begin{tabular}{lll}
\hline
\textbf{Model}                             & \textbf{Hyperparameter} & \textbf{Value} \\ \hline
\multirow{3}{*}{\textbf{Regression trees}} & complexity parameter    & 0.0005         \\
                                           & minsplit                & 20             \\
                                           & minbucket               & 7              \\ \hline
\multirow{4}{*}{\textbf{Random forests}}   & mtry                    & 5              \\
                                           & splitrule               & ``variance"     \\
                                           & min.node.size           & 10             \\
                                           & num.trees               & 200            \\ \hline
\multirow{7}{*}{\textbf{XGBoost}}          & nrounds                 & 400            \\
                                           & max\_depth              & 4              \\
                                           & eta                     & 0.3            \\
                                           & gamma                   & 1              \\
                                           & colsample\_bytree       & 0.8            \\
                                           & min\_child\_weight      & 1              \\
                                           & subsample               & 1              \\ \hline
\multirow{7}{*}{\textbf{Neural networks}}  & number of layers        & 2              \\
                                           & number of neurons       & {[}64, 32{]}   \\
                                           & batch\_size             & 64             \\
                                           & weight\_decay           & 0.0001         \\
                                           & dropout                 & 0\%            \\
                                           & learning rate           & 0.001          \\
                                           & loss                    & ``huber"        \\ \hline
\end{tabular}
\caption{Optimal hyperparameters of tested models}
\label{tab:optimhyparam}
\end{table}

\subsection{Proof of Proposition \ref{prop_existence}}
\label{sec_append_proof1}

Showing the result for $\tau=0$ is sufficient, since a higher value of $\tau$ reduces the expected utility of the traditional insurance contract.

If $\tau=0,$ condition (\ref{cond_achat}) is
$$E\left[\exp(\alpha[m_Y(\alpha|\mathbf{W})-\phi_{\beta}(\mathbf{W})])\right]\leq \exp\left(\alpha [1+\theta_Y-\beta-\beta\theta]E[Y]\right).$$
So taking the logarithm, we get
$$\log E\left[\exp(\alpha[m_Y(\alpha|\mathbf{W})-\phi_{\beta}(\mathbf{W})])\right]\leq \alpha [1+\theta_Y-\beta-\beta\theta]E[Y].$$
Finally,
\begin{equation}\label{partiel}\theta \beta \leq 1-\beta +\theta_Y-\frac{\log E\left[\exp(\alpha[m_Y(\alpha|\mathbf{W})-\phi_{\beta}(\mathbf{W})])\right]}{\alpha E[Y]}.
\end{equation}
Since $\theta\beta$ needs to be strictly positive, this requires
$$\frac{\log E\left[\exp(\alpha[m_Y(\alpha|\mathbf{W})-\phi_{\beta}(\mathbf{W})])\right]}{\alpha E[Y]}< 1-\beta +\theta_Y.$$
Since
$$\log E\left[\exp(\alpha[m_Y(\alpha|\mathbf{W})-\phi_{\beta}(\mathbf{W})])\right]< \sup_{\mathbf{w}\in \mathcal{W}}\frac{m_Y(\alpha|\mathbf{w})-\phi_{\beta}(\mathbf{w})}{E[Y]},$$
it follows from (\ref{cond_bound}) that this condition holds.
Then, taking $\theta\leq \eta(\alpha, \beta)\beta^{-1}$ ensures that (\ref{partiel}), hence (\ref{cond_achat}), holds.

\subsection{Proof of Proposition \ref{prop_solv1} and \ref{prop_solv2}}
\textbf{Proof of Proposition \ref{prop_solv1}.}

We have
$\pi^*_{\phi_{\beta}}=\beta E[Y],$ and 
$$\sigma^2_{\phi_{\beta}}=\beta^2 Var\left(E[Y|\mathbf{W}]\right)=\beta^2 \sigma^2.$$
From Corollary \ref{cor_extension}, which requires that $\theta\leq \eta \beta^{-1},$
$$n\geq n_0=N \mu\left((0,\alpha_0+h_{\beta}(\tau)]\right).$$
Condition (\ref{solvency}) holds if
$$\frac{n_0^{1/2}\theta\pi^*_{\phi_{\beta}}}{\sigma_{\phi_{\beta}}}\geq S^{-1}(\varepsilon).$$
This rewrites
$$\theta\geq \frac{\sigma S^{-1}(\varepsilon)}{N^{1/2} E[Y] \mu\left((0,\alpha_0+h_{\beta}(\tau)]\right)^{1/2}}.$$
To be compatible with the upper bound on $\theta,$ we need condition (\ref{magic}).

\textbf{Proof of Proposition \ref{prop_solv2}.}
\label{proof_solv}

The proof is similar to the proof of Proposition \ref{prop_solv1}, but with (\ref{nombre2}) replaced by (\ref{cond2}).
If $$\theta\geq \frac{a\sigma(1-\varepsilon'^{\gamma})}{\gamma\varepsilon'^{\gamma}},$$
we have
$$S^{-1}\left(\varepsilon - \frac{1}{\left(1+\frac{\gamma \theta}{a\sigma}\right)^{1/\gamma}}\right)\leq S^{-1}(\varepsilon-\varepsilon').$$ Again, we need the upper and lower bounds on $\theta$ to be compatible, which leads to the result of the Proposition.

\subsection{Risk aversion heredity property}

\begin{lemma}
\label{lemma_heredity}
Assume that condition (\ref{cond_achat}) holds for $\phi,$ $\alpha_0>0$ and with $\tau=0.$ Let $F(\alpha)=\Psi'(\alpha)=E\left[Y\exp(\alpha Y)\right].$
Then, (\ref{cond_achat}) holds for all $\tau$ and $\alpha \in (0,\alpha_0+h(\tau)]$ with
\begin{eqnarray*}
h(\tau)=F^{-1}\left(F(\alpha_0)\exp(-\tau)\exp(-\alpha[\pi_Y-\pi_{\phi}])\right)-\alpha_0.
\end{eqnarray*}
\end{lemma}

\begin{proof}
Since $\Psi(\alpha')\geq 1=\Psi(0),$ if (\ref{cond_achat}) holds for $\tau=0,$ it also holds for all $\tau>0.$ Hence we show the result in two steps. First, we show that condition (\ref{cond_achat}) holds for $\tau=0$ and all $\alpha\leq \alpha_0.$ Then we study the case $\alpha_0<\alpha\leq \alpha_0+h(\tau).$

\textbf{First case: $\alpha\leq \alpha_0.$}

Let $Z=Y-\phi(\mathbf{W})-\pi_Y+\pi_{\phi}.$
Condition (\ref{cond_exp}) rewrites
$$E\left[\exp(\alpha Z)\right]\leq \Psi(\alpha').$$
We have, as a consequence of Jensen's inequality,
$$\frac{\log E\left[\exp(\alpha Z)\right]}{\alpha}\leq \frac{\log E\left[\exp(\alpha_0 Z)\right]}{\alpha_0}.$$
Hence
$$ E\left[\exp(\alpha Z)\right]\leq 1.$$

\textbf{Second case: $\alpha_0<\alpha\leq \alpha_0+h(\tau).$}

Let $\alpha=\alpha_0+x$ for $0<x\leq h,$ and $\alpha'=\alpha (1-\exp(-\tau)).$
From a Taylor expansion,
$$\Psi(\alpha')\geq \Psi(\alpha'_0)+x(1-\exp(-\tau))E\left[Y\exp\left(\alpha_0'Y\right)\right].$$
On the other hand,
$$E\left[\exp(\alpha Z)\right]\leq E\left[\exp(\alpha_0 Z)\right]+x \exp(-\alpha\{\pi_{Y}-\pi_{\phi}\})E\left[Y\exp(\alpha Y)\right].$$

Since $E[\exp(\alpha_0Z)]\leq 1\leq \Psi(\alpha'_0),$, condition (\ref{cond_achat}) holds if $$\frac{F(\alpha)}{F(\alpha_0)}\leq \exp(-\tau)\exp(-\alpha[\pi_Y-\pi_{\phi}]).$$ By definition, this condition holds for $h\leq h(\tau).$
\end{proof}

\bibliographystyle{abbrvnat}

\end{document}